\documentclass[english]{article}
\usepackage{lmodern}
\usepackage[T1]{fontenc}
\usepackage[utf8]{inputenc}%
\usepackage{babel,csquotes}
\usepackage[tracking=true, letterspace=50, expansion=false]{microtype}
\usepackage[margin=1in]{geometry}
\usepackage{setspace}
\usepackage{threeparttable}
\usepackage{booktabs}
\usepackage{multirow}
\usepackage{graphicx}
\usepackage{subcaption}

\usepackage[title]{appendix}
\usepackage{amsmath}
\usepackage{amsthm}
\newtheorem{lemma}{Lemma}
\DeclareMathOperator{\var}{var} %
\DeclareMathOperator{\tr}{tr} %
\DeclareMathOperator{\cov}{cov} %
\DeclareMathOperator{\diag}{diag} %

\newcommand{\overbar}[1]{\mkern 1.5mu\overline{\mkern-1.5mu#1\mkern-1.5mu}%
  \mkern 1.5mu}%
\usepackage{xcolor}

\usepackage{etoolbox} %
\AtBeginEnvironment{appendices}{%
  \crefalias{section}{appendix}
  \crefalias{subsection}{appendix}
  \counterwithin{figure}{section}
  \counterwithin{table}{section}
}
\usepackage[pdfusetitle,%
  bookmarks=true,
  bookmarksnumbered=false,
  bookmarksopen=false,
  breaklinks=false,
  pdfborder={0 0 1},
  pdfborderstyle=,
  linkcolor={red!50!black},
  backref=false,
  colorlinks = true,%
  citecolor={black},
  urlcolor={blue!80!black}]{hyperref}
\usepackage{cleveref}
\crefname{appsec}{appendix}{appendices}

\usepackage[
bibencoding=utf8,%
maxbibnames=5, %
minbibnames=1, %
backend=biber,%
sortlocale=en_US,%
style=authoryear,%
uniquename=false,%
url=true,%
sortcites=false,
doi=true,%
eprint=true%
]{biblatex}
\title{Leniency Designs: An Operator's Manual\thanks{We thank Joan Farre-Mensa, Brigham Frandsen, Jeffrey Kling, Emily Leslie, Alexander Ljungqvist, Timothy Taylor, and Heidi Williams for comments. Aurel Rochell provided excellent research assistance.}}
\author{Paul Goldsmith-Pinkham, Peter Hull, and Michal Koles\'{a}r\thanks{Paul Goldsmith-Pinkham is Associate Professor of Finance, Yale School of Management, New Haven, Connecticut. Peter Hull is Professor of Economics, Brown University, Providence, Rhode Island. Michal Koles\'{a}r is Professor of Economics, Princeton University, Princeton, New Jersey.
Their email addresses are \texttt{paul.goldsmith-pinkham@yale.edu}, \texttt{peter\_hull@brown.edu}, and \texttt{mkolesar@princeton.edu}.}}
\date{\today}

\begin{document}
\maketitle

High-stakes decisions are often made by experts: doctors decide whether to treat patients, bail judges decide whether to release defendants before trial, patent examiners decide whether to grant patents to firms, child welfare investigators decide whether to place children in foster homes, and loan officers decide whether to approve consumer loan applications. Usually, even when groups of expert decision-makers agree in many situations, there will be close calls in which systematic disagreement emerges on the appropriate course of action: some decision-makers will tend to be systematically more \emph{lenient}, in that they tend to grant some treatment more often, while others will be more strict.  Furthermore, which decision-maker is assigned to a given case is often either deliberately random (to ensure fairness) or as-good-as-random, due to idiosyncrasies in the assignment process (such as rotations in shifts).

A \emph{leniency design}, also known as a judge or examiner instrument design,  harnesses such exogenous variation to estimate the causal effects of the high-stakes decisions by using a measure of the decision-makers' leniencies as an instrument for the treatment in an instrumental variables regression. Such designs have exploded in popularity in recent years. Prominent applications include \textcite{dobbie2015debt,dobbie2018effects,dobbie2021measuring, farre2020patent, autor2010temporary, doyle2007child, mullainathan2022diagnosing} (see Table 1 of \textcite{frandsen2023judging} for more examples).

Leniency designs rest on a straightforward idea: if decision-makers are randomly assigned, and if assignment only affects outcomes through the treatment decisions, dummies for assignment to different decision-makers are valid instruments. Specifically, we should be able to use the two-stage least squares estimator to estimate the treatment effect. This estimator aggregates the multiple dummy instruments into a single instrument, given by the fitted values from the first-stage regression of treatment on the assignment dummies. With no controls, the fitted values correspond simply to the sample treatment-granting propensity---or sample \emph{leniency}---of each decision-maker. The treatment effect estimate is then calculated as the ratio of the sample covariance of this leniency measure with the outcome relative to its covariance with the treatment.

While intuitive, this simple idea turns out to have many practical complications. The two-stage least squares estimator works well when there are only a few decision-makers. But typically there are many, which makes the sample leniency a noisy measure of its population analog. This noise causes the two-stage least squares estimator produce estimates that replicate some of the selection bias present in ordinary least squares estimates. What is a more robust way of aggregating the assignment dummies---a better leniency measure---that avoids such bias? Other practical questions abound: what controls should we include to ensure credible treatment effect estimates? Should we cluster standard errors, and if so, at what level? Does the fact that leniency is estimated matter for inference? And when can we be confident our estimates are meaningful when treatment effects are heterogeneous?

This paper develops a step-by-step guide to leniency designs that answers these and other questions, complementing a recent review by \textcite{chynexaminer}. We show how the unbiased jackknife instrumental variables estimator (UJIVE) of \textcite{kolesar13late} is purpose-built for leniency designs, leveraging leniency measures that can ensure bias-free causal estimates even in the presence of many decision-makers or controls. We further explain how one can interpret UJIVE estimates under treatment effect heterogeneity, using the local average treatment effect framework of \textcite{ImAn94}. UJIVE is not only useful for estimating treatment effects in this framework: building on \textcite{abadie02} and \textcite{kitagawa15}, we show how it can also be used to test a key identifying assumption of \emph{average monotonicity}.
Finally, we show how knowledge of the decision-maker assignment process---or ``design''---can guide the choice of controls and standard error calculations, and that non-clustered standard errors are often called for.

We conclude with a practical checklist for estimating treatment effects, assessing key assumptions, and probing external validity, all with a unified UJIVE approach. We illustrate this checklist with a re-analysis of \textcite{farre2020patent}, who use quasi-random patent examiner assignment to estimate the value of patents to startups.  We confirm their overall finding that patent-granting significantly increases future patent applications, approvals, and citations on both the extensive and intensive margins, though the magnitude and precision of these estimates are sensitive to using the more robust UJIVE estimator.

\section*{Motivating Example: Estimating the Value of Patents}

We start with the simplest possible version of a leniency design, with
two decision-makers who are completely randomly assigned. Concretely, consider a
stylized version of the setting in \textcite{farre2020patent}, who estimate the causal effect of patent-granting to startup firms on the later inventiveness of those firms. We observe a set of $n$ applications, indexed by $i$, which are
submitted by firms to the US Patent Office. The applications are then assigned
by a random coin flip to one of two examiners, $s$ and $t$. The examiners decide
whether to grant the patent, as indicated by the dummy variable $x_i\in\{0,1\}$.
We refer to $x_i$ as \emph{treatment}, following standard causal inference
lingo. An \emph{outcome} $y_i$ is then realized; for concreteness, suppose $y_i$
measures the number of future patent applications filed by the firm over some
time period (e.g., within five years of submitting the application).

We want to leverage the randomness in examiner assignment to estimate the causal effect of $x_i$ on $y_i$. For now, we assume this effect is the same across all firms: i.e., that being granted a patent causes each firm to submit $\beta$ additional patent applications, compared to being denied a patent. In practice, it is likely that the effects of patent-granting are heterogeneous (i.e., different for different firms); we allow for this possibility later. Ignoring such heterogeneity for now, we write a constant-effects \emph{outcome equation} for $y_i$:
\begin{equation*}
y_i = \gamma + \beta x_i + \varepsilon_i,
\end{equation*}
where the intercept $\gamma$ represents the average number of future patent applications submitted when a firm is denied a patent. The actual number of future patent applications varies across firms due to various unobservable factors captured by the error term $\varepsilon_{i}$.

The core identification challenge here is that estimation of $\beta$ by simple ordinary least squares is likely to suffer from \emph{selection bias}: startups that would produce many patents in the future even if their current application were denied (because, say, they have higher research and development budgets or employ better scientists) may submit stronger applications that are more likely to be granted. This results in a positive correlation between the error term $\varepsilon_i$ and the treatment $x_i$. Consequently, an ordinary least squares regression of the outcome $y_i$ on the treatment $x_i$ will tend to overstate the true causal effect $\beta$.

A leniency design addresses selection bias by leveraging differences in the randomly assigned examiners' tendency to grant patents, isolating variation in the treatment that is unrelated to $\varepsilon_{i}$. Suppose examiner $t$ is ``tough'' (less likely to grant a patent) while examiner $s$ is ``soft'' (more likely to grant a patent). In other words, the overall patent-granting rate of examiner $s$ in the population of patent applications, which we refer to as their \emph{leniency}, is larger than that of examiner $t$: ${p}_s>{p}_t$. Let $\hat{p}_s$ and $\hat{p}_t$ be estimates of these rates: i.e., the fraction of applications we observe each examiner granting in the data. Further, let $z_i$ be a dummy variable for assignment to the soft examiner, so it equals $1$ if application $i$ is handled by examiner $s$ and $0$ if it is handled by $t$. The estimated leniency of the examiner assigned to $i$ is then given by $\hat{\ell}_i=\hat{p}_{t}+(\hat{p}_{s}-\hat{p}_{t})z_i$. A leniency design uses estimated leniency as an \emph{instrument} for patent-granting $x_i$. Specifically, we estimate $\beta$ by an instrumental variable regression, using $\hat{\ell}_i$ to instrument for $x_i$ in the outcome equation.

To gain some intuition as to why this instrumental variable regression can work, note that in this simplified setting---with only two randomly assigned examiners---it is equivalent to a simpler instrumental variable regression which uses the examiner assignment dummy $z_i$ directly as an instrument. In other words, it does not matter whether we use the estimated leniency $\hat{\ell}_i$ as an instrument or the dummy $z_i$: the resulting estimates of $\beta$ are identical. This is because the estimated leniency is just a version of $z_i$ that is scaled by $\hat{p}_{s}-\hat{p}_{t}$ and shifted by $\hat{p}_{t}$, and such instrument scaling and shifting leaves estimates unchanged.

It is easy to see why the equivalent $z_i$-instrumented specification is valid: $z_i$ is randomly assigned, like treatment assignments in a simple randomized controlled trial. Thus, so long as examiner assignment only affects the outcome $y_i$ through the patent-granting treatment $x_i$---the usual \emph{exclusion restriction} for instrumental variable analyses---$z_i$ will be independent of $\varepsilon_{i}$ and hence be a valid instrument for estimating $\beta$.\footnote{Consistency of these estimates also requires a relevance condition: that $z_i$ and $x_i$ have non-zero correlation. Here relevance holds because we assumed that the examiners vary in their leniency, $p_s>p_t$. This means that the population first-stage coefficient from regressing $x_i$ on $z_i$, given by $p_s-p_t$, is non-zero.} Exclusion seems reasonable to assume here, unless we thought examiners did other things besides ruling on the patent (such as giving advice to the firms).

Forming valid standard errors is also easy in this example. Because $z_i$ is randomly assigned by a coin flip for each application, it is uncorrelated across applications. It follows that conventional heteroskedasticity-robust standard errors are appropriate for quantifying uncertainty in the estimate. There is no need to cluster standard errors, analogous to how clustering is not needed in simple randomized controlled trials.

Real-world leniency designs are more complicated than this stylized example, often with many decision-makers who are not assigned completely at random. Handling these features requires some modifications, as we next discuss. Still, the core equivalence between using estimated leniency as an instrument versus using examiner dummy instruments directly will be a useful guide to thinking through estimation and inference in these more complex leniency designs.

\section*{Estimation with Many Decision-Makers and Controls}

The assignment of decision-makers is rarely completely random, but institutional knowledge can sometimes imply it is as-good-as-random once we condition on an appropriate set of control variables. We refer to these as \emph{necessary controls}, as they must be included in every specification leveraging exogenous assignment. As we discuss more below, patent applications in the US are first classified depending on the technology being patented. This determines which group of specialist examiners---the so-called ``art unit''---will review the application. Assignment within art units is then effectively random, conditional on the set of examiners working in the art unit at the time of the assignment. Assuming this set changes slowly over time, a researcher may take art unit-by-cohort fixed effects as the set of necessary controls.\footnote{\label{fn:precision_controls}One may optionally include \emph{precision controls} in some specifications: pre-assignment characteristics help explain variation in the outcome. For example, in the application below, one could include the number of independent claims in the patent application or class group fixed effects. Including these variables will soak up some variability in the error $\varepsilon_{i}$ and thus help reduce the standard errors. This is analogous to including baseline characteristics as controls in regression specifications estimating effects of treatments in randomized controlled trials, while including the necessary controls is analogous to controlling for strata indicators in randomized trials where the treatment assignment probability varies across strata.}

In a realistic setting with many examiners, assignment is now captured by a vector $z_i$ of $K$ instruments: one for each examiner, with one examiner omitted in each art unit to prevent multicollinearity. This leads to a \emph{first-stage equation}, a population regression specification linking the treatment, instrument, and controls:
\begin{equation*}
    x_i = z_i'\pi + w_i'\delta + \nu_i,
\end{equation*}
where the vector $w_i$ comprises art unit-by-cohort dummies and $\nu_i$ is a
regression residual. In the special case where we just have one cohort, so that
$w_i$ reduces to art unit fixed effects, the first-stage coefficients
have a simple interpretation: $\delta_j$ is the overall leniency of the omitted
reference examiner in art unit $j$, and $\pi_k$ measures the difference
between the leniency of examiner $k$ and the reference examiner in $k$'s
art unit.

If we knew the first-stage coefficients $\pi$ and $\delta$, we could use the population first-stage fitted values $\ell_{i} = z_i'\pi + w_i'\delta$ as our leniency measure. By usual partialling-out logic (i.e., the Frisch-Waugh-Lovell theorem), an instrumental variable (IV) regression using this instrument while controlling for $w_i$ is equivalent to running an IV regression without any controls, but instrumenting with the residual from projecting the fitted values onto the covariates, $\tilde{\ell}_i=\tilde{z}_i'{\pi}$. Here $\tilde{z}_i$ denotes the sample residual from regressing the instrument vector $z_i$ onto the controls $w_i$. This leads to the estimator:
\begin{equation*}
   \hat{\beta}^*=\frac{\frac{1}{n}\sum_i y_i\tilde\ell_i}{\frac{1}{n}\sum_i x_i\tilde\ell_i},
\end{equation*}
the ratio of sample covariances (because $\frac{1}{n}\sum_i \tilde{\ell}_i=0$ by construction of the sample residuals $\tilde{\ell}_i$) between \emph{relative leniency} $\tilde{\ell}_i$ and the outcome $y_i$ versus the treatment $x_i$. We use the term relative leniency to stress that $\tilde{\ell}_i$ measures the leniency of examiner handling application $i$ \emph{relative} to other examiners who could have handled the application. In contrast, $\ell_i$ is a measure of \emph{absolute leniency}. With just one cohort, these measures take a simple form: $\ell_{i}$ is the overall patent-granting propensity of $i$'s examiner, while $\tilde{\ell}_i$ is the examiner's patent-granting rate minus the overall patent-granting rate of $i$'s art unit.

The estimator $\hat\beta^*$ is approximately unbiased for $\beta$, because $\tilde{\ell}_i$ is a valid instrument when examiner assignment is as-good-as-random given the controls $w_i$; \Cref{sec:derivations} gives the formal derivation of this claim along with other results below. In practice, however, it is infeasible to use this true relative leniency measure as an instrument because we do not know the first-stage coefficients. A tempting alternative is to simply use least squares estimates of the first-stage coefficients, $\hat{\pi}$ and $\hat{\delta}$, in place of the unknown population values. This is exactly equivalent to using a two-stage least squares (2SLS) estimator which instruments with the full set of assignment dummies $z_i$. To see this, recall that the 2SLS estimator is equivalent to an instrumental variables estimator that uses a single instrument given by the sample first-stage fitted values, $z_i'\hat{\pi}+w_i'\hat\delta$. By the same partialling-out logic that led to the equation for $\hat{\beta}^*$, this is in turn equivalent to an IV regression without covariates that replaces the population relative leniency $\tilde\ell_i$ in the formula for $\hat\beta^*$ with the sample analog $\hat{\ell}_i=\tilde{z}_i'\hat{\pi}$ (the residual from projecting the fitted values onto the controls).

However, it turns out using the estimated relative leniency $\hat{\ell}_i$ will tend to produce biased treatment effect estimates when there are many examiners. This is the classic many-weak instrument bias problem for 2SLS \parencite[e.g.,][]{bekker94,bound1995problems}; it arises here because $\hat{\ell}_i$ is constructed in part from the treatment status of observation $i$ through the least squares estimate $\hat{\pi}$. With just one cohort, for instance, $\hat{\ell}_i$ is given by the fraction of patents granted by the examiner assigned to $i$ in the sample minus the fraction of patents granted by the art unit handling the application---and both fractions are computed including the data from application $i$. But applicant $i$'s treatment likely correlates with the outcome error $\varepsilon_i$; this is, after all, the reason to use an instrumental variables instead of an ordinary least squares regression. Thus, the estimated leniency instrument is also likely correlated with the outcome error, and this generates bias in the instrumental variable estimates in the direction of the naïve ordinary least squares estimates.

The bias from using the estimated leniency instrument can be severe, especially when examiner assignment only modestly affects the probability of treatment.  A helpful rule of thumb for gauging the magnitude of the bias obtains when we assume the errors are homoskedastic; then, the 2SLS bias approximately equals the bias of ordinary least squares divided by $ (1-R^2)\times E[F]$, where $E[F]$ is expectation of the first-stage $F$ statistic for the hypothesis that the first-stage instruments are irrelevant (i.e., that $\pi=0$ in the first stage equation), and $R^2$ is the population partial R-squared from adding instruments to the first-stage regression. In practice, since $R^2$ tends to be near zero, the magnitude of $E[F]$ tells us how much smaller 2SLS bias is relative to ordinary least squares (it's never larger, since $E[F]$ always exceeds one). In fact, this relationship can be used to justify the popular $F>10$ rule of thumb proposed by \textcite{staiger1997instrumental} for identifying weak instruments. If the expectation of the first-stage $F$ statistic lies below 10, then 2SLS bias exceeds 10\% of least squares bias, so the rule of thumb can be thought of as a diagnostic for whether we are in this large bias region.\footnote{\Textcite{StYo05} construct a formal statistical test based on $F$ for the null hypothesis that the 2SLS bias is large in the sense of potentially exceeding 10\% of least squares bias, under homoskedastic errors. While the precise critical value depends on $K$, it is close to 10 for most values of $K$: see their Table 5.1.} With many examiners, the first-stage $F$ statistic can be modest even when examiner leniency explains economically meaningful variation in the treatment because the formula for $F$ divides by $K$.

Knowing that the 2SLS bias comes from using own treatment status to estimate the
relative leniency suggests a natural bias-free alternative: instrument with a
leniency estimate $\hat{\ell}_{-i}$ that \emph{leaves out} observation $i$'s own
value of the treatment, thereby avoiding a mechanical correlation with
$\varepsilon_i$. The unbiased jackknife instrumental variable estimator (UJIVE),
studied in \textcite{kolesar13late}, implements this logic by setting the
relative leniency estimate to $\hat{\ell}_{-i}=\tilde{z}_i^\prime\hat\pi_{-i}$
where $\tilde{z}_i$ are the same instrument residuals as before and
$\hat\pi_{-i}$ is a least-squares estimate of $\pi$ which uses all observations
except for $i$. Because this leave-one-out---or ``jackknifed''---estimate is
unbiased for $\pi$ and, in \emph{iid} data, independent of the data for
observation $i$, the same arguments used to show approximate unbiasedness of the
infeasible estimator $\hat\beta^*$ can also be used to show approximate
unbiasedness of the UJIVE estimator. Moreover, as shown in  \Cref{sec:derivations},
UJIVE is simple to run: the vector of
relative leniencies can be calculated in a single step, using some matrix algebra tricks; we do not need to run $n$ separate leave-out regressions to get the $\hat\pi_{-i}$'s.

The unbiasedness property of UJIVE hinges on the fact that it uses leave-out
estimation to directly estimate relative leniency $\tilde{\ell}_i$ accounting
for controls through the residualized $\tilde{z}_i$. In contrast, leave-out
estimation of absolute leniency $\ell_i$ tends to work poorly when the number of
covariates is large. In particular, \textcite{PhHa77,aik99} propose the
jackknife IV estimator (JIVE) that constructs a leave-out estimate of the
absolute leniency as an instrument in a regression with covariates
(\textcite{aik99} term this estimator JIVE1; the related JIVE2 estimator has
similarly poor performance with many covariates). By the same partialling-out logic as before, this is equivalent to using 2SLS without controls, and using as an
instrument the least squares residuals from a regression of the leave-out
absolute leniency estimates on the controls. Even though the JIVE absolute leniency
estimates are by construction uncorrelated with own treatment, the covariate
adjustment reintroduces own-observation bias because the residuals depend on
absolute leniency estimates for \emph{other} observations---which \emph{are}
constructed using one's own treatment status.

To see this problem of classic JIVE simply, consider the case of just one cohort so the covariates consist of
art unit fixed effects. The JIVE absolute leniency estimate of application $i$
is then the average leave-out patent-granting rate of examiner assigned to $i$. We
 regress these on art-unit fixed effects using ordinary least squares, and
take the residuals. The fitted value is the average of the leave-out
patent-granting rates in the art unit assigned to $i$. Crucially, because
leave-out patent-granting rates for all other cases handled by $i$'s examiner
depend on the treatment of application $i$, the average of the leave-out
rates is the same as the overall (non-leave-out) average rate of
the art unit (the average leniency of all examiners in the art unit).

Specifically, consider an examiner $j$ working in a given art unit, who handles
$n_{j}$ cases $\mathcal{J}$. The sum of the leave-out rates for this examiner is
$\sum_{i\in\mathcal{J}}\left(\frac{1}{n_{j}-1}(\sum_{i'\in\mathcal{J}}x_{i'}-x_{i})\right)=
\frac{1}{n_{j}-1}\left(n_{j}\sum_{i'\in\mathcal{J}}x_{i'}-\sum_{i\in\mathcal{J}}x_{i}\right)=\sum_{i\in\mathcal{J}}x_{i}$---just
the overall number of cases granted by $j$. Consequently, summing leave-out rates
for all examiners working in the art unit and dividing by the number of total
cases gives the average patent-granting rate of the art unit. Thus, the
residualized JIVE leniency takes the leave-out patent-granting rate of examiner
assigned to $i$, and subtracts the overall patent-granting rate of the art unit,
which depends on the treatment status of application $i$. This reintroduces the
own-observation bias of 2SLS that we were trying to get rid of, except that it
typically runs in the opposite direction to 2SLS (since the own-observation
component is now only the term being subtracted off, if $i$ is treated, this
\emph{decreases} JIVE's relative leniency measure). Under homoskedasticity, the
JIVE bias approximately equals that of 2SLS multiplied by
$-E[F] \times L/((E[F]-1)K-L)$ where $L$ is the number of controls. The JIVE
bias is negligible when the number of controls is much smaller than the number
of instruments (in fact, UJIVE and JIVE coincide in the absence of controls and
a constant). But when many controls (e.g., art unit-by-cohort fixed effects) are
needed to ensure as-good-as-random assignment, this
formula shows that the magnitude of JIVE bias can be comparable to that of
2SLS.

Two more comments  are warranted here. First, it is common practice to
report the first-stage $F$ statistic with 2SLS estimates; as discussed, its
expectation $E[F]$ informs 2SLS bias. But small first-stage $F$
statistics need not worry a researcher using UJIVE\@: the arguments for its
approximate unbiasedness work even if $E[F]$ is small. In fact, UJIVE remains
approximately unbiased and consistent even when the instruments are weak enough
that $E[F]$ converges to one in large samples, so long as $\sqrt{K}\times (E[F]-1)$
is large \parencite[see, e.g.,][]{MiSu22}. Second, the leave-out estimation
approach leveraged by UJIVE assumes independence across observations. When
observations are instead clustered together---a scenario we detail when
discussing the calculation of standard errors---a leave-own-cluster-out
modification of UJIVE may be needed to ensure unbiasedness
\parencite{frandsen2025cluster,kmwz26}. Intuitively, correlations between
treatment $x_i$ and errors $\varepsilon_j$ for observations $i$ and $j$ in the
same cluster will generally reintroduce the mechanical bias between the
leave-own-observation-out UJIVE instrument and the errors, again tending to bias
estimates towards ordinary least squares. Hence, for examiner designs,
clustering in the data is not just a consideration for inference but also guides the choice of estimator.

\subsection*{Alternatives to UJIVE}

While UJIVE is a natural solution to the 2SLS bias problem, there are other
reasonable approaches. \Textcite{AcDe09} propose a clever modification of
JIVE---the improved jackknife IV estimator (IJIVE)---which can greatly reduce
its bias in the presence of controls by reversing the order of operations:
\emph{first} residualize the instruments and \emph{then} compute a leniency
measure based on leave-out fitted values. While reversing the order of
operations doesn't entirely eliminate the own-observation bias (because the
initial residualization is not leave-one-out), in practice IJIVE and UJIVE tend
to produce similar estimates. More recently, \textcite{csw23} propose an
alternative to UJIVE, termed FEJIV\@. We show in the appendix that FEJIV can be
interpreted as solving for a relative leniency measure with the smallest mean
squared error under homoskedasticity, subject to two constraints: that the resulting
estimator is free of own-observation bias and that the leniency measure is
orthogonal to the covariates. The UJIVE leniency measure, in contrast, achieves
the minimal mean squared error without the latter orthogonality constraint. The
orthogonality property of FEJIV is attractive in that adding a linear function
of covariates to the outcome does not affect the estimate; the price for this is
that FEJIV leniency is slightly noisier. The resulting estimator also tends to
be computationally demanding in large datasets, and imposes stronger data
requirements (e.g., the estimator may not exist if there are high-leverage
observations).

Another approach is to bias-correct 2SLS, which again replaces the infeasible instrument $\tilde{\ell}_i$ in the formula for $\beta^*$ with the estimated relative leniency $\hat{\ell}_i$. When 2SLS does this, it increases the expectation of the denominator from $\frac{1}{n}\sum_i\tilde{\ell}_i^2$ to $\frac{1}{n}\sum_i\tilde{\ell}_i^2+K\operatorname{var}(\nu_i)/n$ under homoskedasticity. The quantity $\sum_i\tilde{\ell}_i^2$ is the numerator of the population partial R-squared statistic; it measures the increase in the explained sum of squares from adding instruments to the first stage. Since 2SLS uses in-sample fitted values to estimate the predictive power of the instruments, it overstates this predictive power---exactly analogous to how the unadjusted R-squared overstates the predictive power of regression (the same issue happens in the numerator, and taking the ratio gives the rule-of-thumb bias formula for 2SLS). Paralleling how the adjusted R-squared fixes this issue with a degrees of freedom correction, we can use degrees of freedom adjustments in both the denominator and the numerator of the 2SLS formula. Unfortunately, as with the adjusted R-squared formula, the resulting bias-corrected 2SLS estimator (due to \textcite{nagar1959bias}) only works under homoskedasticity.

In practice, rather than using UJIVE or its cousins that compute the leniency
measures internally, it is common for researchers to first construct an
``external'' leniency measure and then use it as an instrument in a
just-identified IV regression. This is analogous to first computing first-stage
fitted values and then using them as a single instrument, rather than directly
using the standard 2SLS estimator which estimates everything in one step. Such
``manual 2SLS'' procedures are widely advised against \parencite[see,
e.g.,][Chapter 4.2]{angrist2009mostly} and we would similarly advise against
manual leniency IV estimation.\footnote{Likewise, we would advise against interpreting the variance of a manual (or UJIVE-derived) leniency measure as a signal of the design strength as is sometimes done in practice. For one thing, noise in estimation will tend to overstate the true variability in leniency. Properly computed UJIVE standard errors, discussed below, are sufficient for gauging the design's power.} Indeed, as shown above, subtle variations in how
the leniency measure is exactly constructed and accounts for covariates can have
potentially large consequences for the bias of the ultimate estimator. Also,
just as manual 2SLS implementations lead to incorrect second-stage standard
errors, so do manual leniency constructions. It is simpler and safer to use
a one-step implementation like UJIVE, with established unbiasedness properties.\footnote{One justification for such manual leniency measure constructions is that the
researcher faces missing data: there is a large sample of
observations for the first-stage regression, but only a smaller subset with outcomes to estimate effects on. This scenario is easily handled with the UJIVE
approach, which simply computes the leave-out estimate $\hat{\pi}_{-i}$ on the
full sample. Relative to dropping the missing data, this will increase
efficiency with many examiners when the true first-stage coefficients match in the samples with and without outcome
data. Otherwise, dropping the missing data may be preferred.}

\section*{Heterogeneous Treatment Effects}

So far, we have assumed the treatment parameter $\beta$ is constant. This can of course be a strong assumption in practice. It means, for example, that any successful patent application raises the future innovativeness of all startups $i$ by the same constant amount. In reality, the value of patents is likely higher for some startups than for others. Luckily, the UJIVE estimator can retain a causal interpretation even in the presence of such heterogeneous effects---i.e., when patent effects $\beta_i$ vary arbitrarily across firms $i$.\footnote{In this section only, we assume the treatment is binary. This is mostly for notational convenience, since otherwise effect heterogeneity could come both from differences across observations and from different margins of treatment response for a given observation. See \textcite[Appendix A.2]{kp25} for an extension of the local average treatment effect theorem to the case with multivalued or continuous treatment.}

\subsection*{First-Stage Monotonicity}

The local average treatment effect theorem of \textcite{ImAn94}, recognized in the 2021 Nobel Prize, offers a guide to the general heterogeneous-effect interpretation of leniency designs. The theorem maintains the assumption of as-good-as-random decision-maker assignment and the IV exclusion restriction, while replacing the restriction of constant treatment effects with an assumption of first-stage \emph{monotonicity}.

In the patent setting, monotonicity means that if we can find a case that some examiner $a$ would grant a patent to but some other examiner $b$ would deny, it must be that $a$ is more lenient in \emph{all} cases; there cannot be another case that $a$ would deny but $b$ would grant. In the simple example with just two examiners, this implies we can divide the universe of firms into two groups: a \emph{complier} group with patent applications that are granted when assigned to the soft examiner $s$ but not when assigned to the tough examiner $t$, and a non-responder group whose treatment status is unaffected by examiner assignment (they are always either granted or denied, regardless of which examiner handles the case). Monotonicity rules out the presence of a third group of \emph{defier} firms that are granted their applications only if assigned to $t$. The local average treatment effect theorem states that, in the absence of defiers, the IV regression using an indicator variable for assignment to the soft examiner as an instrument estimates the average treatment effect for compliers. \Textcite{ImAn94} call this a local average treatment effect, to stress that we learn nothing about treatment effects for non-responders.
In contrast, if we ran a randomized controlled trial that granted applications by a coin flip, we would learn the overall average treatment effect.

With many as-good-as-randomly assigned decision-makers, monotonicity implies there are many complier groups (one for each pair of decision-makers). An extension of the theorem shows that using the relative leniency $\tilde{\ell}_i$ as an instrument---or the unbiased leniency measure constructed by UJIVE---identifies a weighted average of complier-group specific treatment effects, weighted using the squared leniency differences of each decision-maker pair (\Cref{sec:derivations} gives the formal result). This follows because, in the population, leniency IV is equivalent to running a series of simple IV regressions with a single binary instrument, that restricts the sample to a given decision-maker pair, then averaging these IV regressions together using squared pairwise leniency differences as weights. Since a leniency difference measures the share of compliers, leniency IV thus weights by the square of the group size. Hence, larger complier groups receive more weight. But importantly the weighting is convex: no complier groups are given negative weights.

\subsection*{Weakening Monotonicity}

While it is likely more palatable than assuming constant treatment effects, monotonicity is nevertheless a strong assumption. In particular, results in \textcite{vytlacil02} imply it is equivalent to assuming each patent examiner has the same ranking of the  patent applications' merits, with examiners only differing in the cutoff they use for gauging whether an application is above the bar. The strength of this assumption for leniency designs was, in fact, first noted in the original \textcite{ImAn94} analysis (see their Example 2, p.~472). There are two core issues. First, in many settings like our patent example, there are multiple dimensions to decision-makers' evaluation criteria which may make it unlikely that any two decision-makers would have the same ranking. For instance, patent applications are evaluated by their perceived usefulness, novelty, and non-obviousness, and different examiners likely place different weights on these criteria. Second, even if the weights were the same, variation in skill among decision-makers can induce ranking differences due to mistakes by less-skilled decision-makers. For instance, \Textcite{cgy22} show that in the case of radiologists, variation in skill accounts for about 40\% of variation in their leniency.

Fortunately, the usual first-stage monotonicity condition can be weakened in three ways. To see how, first consider what IV identifies without monotonicity. In general, a simple IV regression restricted to a pair of decision-makers identifies a weighted \emph{difference} between complier and defier treatment effects with weights given by the fraction of compliers and defiers (because defier treatments are shifted in the opposite direction by decision-maker assignment, relative to the compliers). Since leniency IV averages these pairwise IV regressions, it can be seen as identifying a weighted-average treatment effect for compliers and defiers, but with negative weight put on all defier groups. This is a problem for its causal interpretation, in general. For example, we risk ``sign reversals'': even if patent-granting, say, uniformly increases future firm productivity so that $\beta_i$ is positive for all $i$, the leniency IV could estimate a negative effect---namely, if the treatment effect for defiers is much larger than that for compliers. But while the presence of defiers necessarily causes a negative weighting problem in the case with a single pair of examiners, there is a subtlety with multiple examiners. Since each firm appears in multiple pairwise examiner comparisons, the total weight we place on each firm corresponds to its average complier, or defier, status across all pairwise comparisons. 

This logic suggests the first way the usual monotonicity assumption can be weakened: as long as no firm is a defier ``on average'' the weights placed on individual firms remain positive. Equivalently, one may assume that the average leniency of the examiners who would grant the firm a patent exceeds the average leniency of those who would deny it---a condition \textcite{frandsen2023judging} term ``average monotonicity.''\footnote{While the statement of the condition in \textcite{frandsen2023judging} is more complex, we show in the appendix that our formulation is equivalent. With two decision-makers, average monotonicity reduces to the usual monotonicity condition.}

The second way comes from seeing that even if some weights are negative, it only presents a problem when the patent applications with negative weights systematically differ in their treatment effects. One may be willing to assume this is not the case.\footnote{\label{fn:selection_on_gains}\Textcite{heckman2006understanding} call this assumption no ``essential'' treatment effect heterogeneity. In terms of the \textcite{roy1951some} selection model, it imposes no ``selection-on-gains.''} More generally, \textcite{dechaisemartin17} points out that if one finds a subset of positively weighted applications that matches the treatment effects of the firms that are defiers on average, these compliers cancel out the negative treatment weights on the defiers; we can then interpret the IV as estimating a convex weighted average of treatment effects for the remaining compliers.

The third point is that even if the negatively weighted observations differ systematically, this need not bias IV estimates substantively unless there are many firms with non-negligibly negative weights and the systematic differences in treatment effects are very large. We never see two patent examiners evaluate the same case, so we cannot directly assess the common ranking assumption. But we do actually have some evidence in the context of judge assignment, because some court cases are assigned to a panel of judges. There \Textcite{sigstad_monotonicity_2025}  shows that while monotonicity is technically often violated in judicial panels, the ranking disagreements are not severe enough to create substantial bias in leniency IV estimates.

\subsection*{Testing Monotonicity}

Typically, leniency IV designs do not feature multiple decision-makers being assigned to the same case, precluding such direct monotonicity tests in other contexts. But we can still indirectly test the assumption. One implication of monotonicity is that the average outcomes for cases assigned to two decision-makers with the same leniency must be the same in the population, because under monotonicity, the two decision-makers' decisions must exactly match. More generally, if the leniencies are similar, the average outcomes must also be similar since the number of observations on which the decision-makers disagree must be small.
\Textcite{frandsen2023judging} formalize this idea in a statistical test.\footnote{Strictly speaking, both this test and the average monotonicity test described below are of the joint null hypothesis of as-good-as-random assignment, exclusion, and monotonicity, as a violation of any assumption could drive a test rejection.} While useful in some settings, the test has two limitations. First, to show its validity, \textcite{frandsen2023judging} rule out a large number of decision-makers (i.e., the original motivation for using UJIVE), and its implementation requires bounded outcomes (ruling out, e.g., looking at patent effects on future startup profits). Second, it tests the original \textcite{ImAn94} monotonicity condition, not the weaker average monotonicity condition that is necessary and sufficient for nonnegative weights.

To address both limitations, we propose here a test for average monotonicity---based on an earlier proposal by \textcite{kitagawa15}---which leverages the UJIVE estimator.\footnote{An alternative informal test discussed in \textcite{frandsen2023judging}, and used previously in the literature \parencite[e.g.][]{dobbie2015debt,dobbie2017consumer}, is to check that leniency rankings of decision-makers are stable across subsamples defined by baseline characteristics. For example, subgroup regressions of treatment on relative leniency should yield a positive coefficient in each subgroup. Establishing and comparing the power of such tests to our proposal is an interesting area for future research.}
We develop our proposed monotonicity test by first noting, following \textcite{abadie02}, that the UJIVE estimator can not only estimate average treatment effects for compliers; it can also \emph{characterize} compliers by their baseline characteristics and potential outcomes. Formally, consider a UJIVE estimator with the same treatment, instruments, and controls as before, but with a modified outcome. Rather than $y_i$, we put on the left-hand side $\tilde{y}_i=v_i\times x_i$, where $v_i$ equals some variable determined before the as-good-as-random assignment of $z_i$. The local average treatment effect theorem applies to this modified outcome as well, showing that under average monotonicity, the UJIVE estimator identifies a convex weighted average of treatment effects of $x_i$ on $\tilde{y}_i$. But by construction these ``effects'' are just $v_i$, since $\tilde{y}_i$ is moved by exactly this amount when $x_i$ is counterfactually increased by one unit. Moreover, the weights that UJIVE puts on these effects are exactly the same as the ones in the original treatment effect specification. Hence, by simply replacing $y_i$ with $\tilde{y}_i$, we can easily compute weighted averages of $v_i$ with the UJIVE weights.

To see why this result is useful for testing monotonicity, note that it is possible to obtain a logically invalid estimate when computing such weighted averages. If $v_i$ is a binary variable, which can only take on values zero or one, then any convex weighted average of $v_i$ must not be negative or greater than one. Such a finding would suggest something has gone wrong in the local average treatment effect theorem, namely monotonicity (assuming we are confident in as-good-as-random assignment and exclusion). More generally, we can set $v_i$ to be an indicator that some baseline variable---or set of variables---takes on a particular set of values and check that the resulting UJIVE estimate with that $\tilde{y}_i$ as the outcome lies between zero and one. For binary outcomes, we can even use the original $y_i$ (by setting $\tilde{y}_i=y_i\times x_i$) since then UJIVE will estimate a weighted average of \emph{treated outcomes} \parencite{abadie02}; weighted averages of untreated outcomes can be estimated by setting $\tilde{y}_{i}=y_i\times (x_i-1)$.
Of course, monotonicity tests like these are as straightforward to implement via UJIVE, and they inherit its approximate unbiasedness property even with many decision-makers or controls.

In practice, this test will likely only detect gross violations of monotonicity: to move a given weighted average beyond the logical $[0,1]$ bounds we need on-average defiers to be both prevalent enough and different enough from those who are compliers on average. Relative to tests for stronger versions of monotonicity (such as in \textcite{frandsen2023judging}  or the \textcite{chmw24} version of \textcite{kitagawa15}), this test's rejections raise more meaningful concerns for a design's validity via the potential for sign reversals.

\bigskip

We close this section with two additional comments. First, one might wish to use the above method to estimate average baseline characteristics of compliers even when monotonicity is not a concern. This helps evaluate another more subtle concern with IV estimation given heterogeneous treatment effects: the \emph{external validity}, or representativeness, of the estimated complier-average treatment effect. By replacing the $y_i$ outcome with $\tilde{y}_i=v_i\times x_i$ for some baseline characteristic $v_i$, researchers can directly evaluate whether observable complier characteristics are representative of the broader study population.\footnote{Another complementary way of gauging external validity is to report how complier-average treatment effects, computed by simple IV regressions restricted to a given examiner pair, vary with the pairs' leniencies. Ordering the examiners by their leniency and reporting the estimates for pairs of neighboring examiners then gives an estimate of the marginal treatment effect curve \parencite{HeVy05}. Intuitively, a flat curve suggests that effects for non-responders are likely similar to the UJIVE estimate. However, the properties of such procedures---or their more sophisticated versions \parencite[e.g.,][]{mst18}---have not to our knowledge been formalized in the case of many decision-makers or controls.}

Second, a further issue can arise when the specification of controls in $w_i$ is insufficiently flexible. To avoid omitted variable bias with instruments that are only conditionally as-good-as-randomly assigned, the covariates $w_{i}$ need to enter flexibly enough so that the conditional mean of the instrument vector given $w_i$, $E[z_i\mid w_i]$, is linear in the covariates (\textcite{kolesar13late} shows this condition is sufficient, while \textcite{blandhol2022tsls} show it is necessary). This condition holds automatically when $w_i$ is a set of fixed effects, but if other or multiple sets of controls are needed it may require adding interactions and higher-order terms. Further flexibility may then be needed for average monotonicity to hold. As we show in \Cref{sec:derivations} (\Cref{lemma:lambda_weights}), the average monotonicity condition applies to the population leniency measure based on the first-stage equation. If, say, there are examiners working in multiple art units with art-unit-specific leniencies then a first stage that only controls for art-unit-by-cohort fixed effects additively will be misspecified. This misspecification could cause average monotonicity to fail unless we run a first stage that interacts the examiner assignment with art-unit-by-cohort fixed effects \parencite[like the specification considered in][]{angrist1991does}. However, as discussed above, even if uninteracted control specifications could yield negative weights, they need not generate meaningful bias---and will typically yield more precise estimates. In practice, researchers can explore this potential bias-variance tradeoff by checking robustness to more flexible control specifications.

\section*{Standard Errors}

Leniency designs with multiple decision-makers can also raise subtle standard error issues. To see why, we start with a benchmark: suppose we knew the relative leniency $\tilde{\ell}_i$, and hence could compute the infeasible estimator $\hat{\beta}^*$. With \emph{iid} data, the correct standard errors are given by the square root of $\sum_i\hat{\varepsilon}_i^2 \tilde{\ell}_i^2/(\sum_i \tilde{\ell}_i^2)^2$, where $\hat{\varepsilon}_i$ is the residual from projecting $y_i-x_i\hat{\beta}^*$ on the covariates. Since the covariances in the numerator and denominator of $\hat\beta^*$ are approximately normally distributed by the central limit theorem, their ratio is also approximately normal with this standard error. This follows by the delta method, a classic result in asymptotic statistics.
Standard software packages compute heteroskedasticity-robust standard errors for 2SLS using the feasible version of this formula, plugging in $\hat{\ell}_{i}$ for  $\tilde{\ell}_i$ and the 2SLS estimate for $\hat{\beta}^*$. These standard errors are correct when treatment effects are homogeneous and the many instrument 2SLS bias is negligible. Under these conditions, 2SLS is as efficient as the infeasible estimator $\hat\beta^*$.

Unfortunately, with many instruments, the same issue that causes bias in the 2SLS estimator also causes its standard errors to be misleadingly small. Recall from our discussion of bias-corrected 2SLS that when 2SLS plugs in
$\hat{\ell}_{i}$ for  $\tilde{\ell}_i$ in the formula for $\hat{\beta}^*$, it inflates the expectation of the denominator from $\frac{1}{n}\sum_i\tilde{\ell}_i^2$ to $\frac{1}{n}\sum_i\tilde{\ell}_i^2+K\operatorname{var}(\nu_i)/n$ under homoskedasticity.
Since the same denominator appears in the standard error formula (up to the scaling by $n$), overfitting shrinks standard errors too. Thus, 2SLS mimics not only the ordinary least squares estimate, but also its precision. As \textcite{bound1995problems} emphasize, this can mask the weak instrument problem: researchers may see tight standard errors around an estimate like that of ordinary least squares and wrongly think the causal effect is precisely estimated with minimal bias.

UJIVE fixes the denominator problem by construction, but the numerator  of the default heteroskedas\-ticity-robust standard error formula still has two issues.  First, with heterogeneous treatment effects, we cannot match the infeasible estimator's precision even when instruments are strong: the correct numerator picks up a second term reflecting variation in complier treatment effects across complier groups. Conveniently, \textcite{ImAn94} already derived the correct numerator for the standard error in their appendix \parencite[see also][]{lee18late}. Second, with many examiners, a third term appears due to estimation noise in $\hat{\ell}_{i}$ (this is the \textcite{bekker94} many instrument term; we give details in \Cref{sec:derivations}). Luckily, if we use UJIVE along with the plug-in heterogeneity-robust formula, it turns out that this also takes care of the many instrument term---we use this formula in our empirical illustration below.\footnote{The plug-in formula for the standard error's numerator has an own-observation bias that makes it overshoot its target, similar to the upward bias in the denominator of 2SLS\@. By a lucky coincidence, this overshooting more than accounts for the many instrument term making the formula valid (albeit slightly conservative). Correcting the slight upward bias involves a more complicated leave-three-out procedure: see  \textcite{AnSo23,yap25,kmwz26}.}

Unfortunately, there are still two potential issues with UJIVE standard errors. The first is a classic one, dating back to  \textcite{AnRu49}; it concerns the delta method argument that since the covariances in the numerator and denominator of $\hat\beta^*$ are each approximately normal, so is their ratio. For this argument to work well, we need the signal-to-noise ratio in the denominator (the ratio of its mean to its standard deviation) to be relatively large. To see why, suppose that there is no variation in leniency across patent examiners. Then the infeasible estimator $\hat{\beta}^*$ is a ratio of two normals, both with zero mean. Ratios of mean-zero normals follow a Cauchy distribution, which has much thicker tails than a normal distribution. Typical standard errors will not account for this. When the first stage is weak so the signal-to-noise ratio is non-zero but small, the estimator will have tails that are a bit thinner than Cauchy, but still much thicker than normal tails. For the feasible estimator, the signal-to-noise ratio scales with $\sqrt{K}\times (E[F]-1)$; if the first stage is weak enough to make this small, we have a weak instrument problem.

A vast literature has tackled the weak instrument problem over the years, including \textcite{AnRu49} for the case with a fixed number of instruments and several more recent papers proposing jackknife extensions for many instruments \parencite[e.g.,][]{MiSu22,MaOt24,yap25}. The fix turns out to be quite simple. To test that the causal effect equals a particular value $\beta_0$, we compute the residual $\hat{\varepsilon}_{i0}$ from projecting $y_i-x_i\beta_0$ onto the covariates. If the null is true, $\hat{\varepsilon}_{i0}$ should be uncorrelated with relative leniency, so we just need to check whether this correlation is statistically significant. This turns out to be equivalent to computing the UJIVE standard error just like before and checking whether UJIVE is significantly different from $\beta_0$, with one tweak: we use $\hat{\varepsilon}_{i0}$ instead of the UJIVE residual.\footnote{This is the test proposed by \textcite{yap25}. The test proposed by \textcite{MaOt24} uses the same idea, but doesn't use the heterogeneity-robust formula for the standard error. As a result, it is not robust to treatment effect heterogeneity, and neither is the many-instrument robust version of the \textcite{AnRu49} test considered in \textcite{MiSu22}.}

Recently, for the single-instrument case, \textcite{AngKo24} give a more optimistic take on the weak instrument problem: they observe that, for the case with known leniency, the usual standard error only leads to overly optimistic inference if the correlation between $\frac{1}{n}\sum_i \varepsilon_i\tilde{\ell}_i$ (the numerator of $\hat{\beta}^*-\beta$) and $\frac{1}{n}\sum_i x_i\tilde{\ell}_i$ (the denominator) is very high; if not, the weak instruments blow up the standard error with $\hat{\varepsilon}_i$ even more than if we used the robust \textcite{yap25} approach that uses $\hat{\varepsilon}_{i0}$ (which in this case coincides with the Anderson-Rubin test). This correlation parameter can be thought of as a measure of endogeneity, since it corresponds to the correlation between $\varepsilon_i$ and $x_i$ under homoskedasticity. But severe selection biases leading to very high endogeneity are unlikely in many applications; in such cases, the usual standard errors will not mislead. We show in \Cref{sec:derivations} that an analogous argument applies to the UJIVE estimator with many instruments and controls: the only difference is the correlation parameter uses the UJIVE leave-out leniency measure rather than $\tilde{\ell}_i$. While this correlation parameter no longer directly maps to endogeneity due to the more complicated variance structure, it can still be bounded by placing bounds on the treatment effects. In most applications these weak instrument issues are moot, as the variation in leniency (as measured by $\sqrt{K}\times (E[F]-1)$) will be substantial. But if the variation is small and a high correlation parameter cannot be ruled out, using $\hat{\varepsilon}_{i0}$ to compute standard errors with the above weak-instrument test may be prudent.

\subsection*{Clustering}
A final complication arises when the data is not \emph{iid} but clustered. This can happen for two reasons. Either because the sampling is clustered (say the patent office, instead of sharing the full data, only shares application data submitted on a randomly selected subset of dates, and the researcher would like to generalize to the full set of data), or because the assignment is clustered (say there is a lottery for each date, and the examiner who wins the lottery is assigned all applications filed on that date). Then, even if the relative leniency $\tilde{\ell}_i$ was known, a clustered version of the standard error formula is needed. Often, a researcher observes the full population of interest (in our application, we observe all patents in a specific time frame), so the main consideration is the assignment process.

\Textcite{aaiw20,aaiw23} show that if the examiner assignment process is \emph{iid}, one does not need to worry about the correlation patterns of the residual $\varepsilon_i$ across $i$, which is what traditionally motivates the decision to cluster. What matters for the standard error is the correlation pattern of the \emph{product}, $\tilde{\ell}_i\varepsilon_i$. Traditional clustering considerations treat the instruments and hence $\tilde{\ell}_i$ as fixed, or conditioned upon, in which case correlation patterns in the residual do matter for the clustering decision. Under the random instrument view, in contrast, the product will be uncorrelated whenever examiner assignment is independent across $i$, and we can use the same rule of thumb as in randomized experiments: cluster at the level of variation in the assignment. Furthermore, whether clustering matters for the magnitude of the standard error does not help determine whether it is \emph{needed}, so clustering the standard errors ``just in case'' may yield overly conservative inference. This reasoning also applies when leniency is unknown, but the number of examiners is fixed.

While these arguments have not yet been formally extended to the many examiner case, it seems natural to let the assignment process guide clustering decisions. When examiner assignment is independent across cases, robust standard errors suffice. When examiners assignment is instead clustered---such as when doctors are assigned to shifts, and a single doctor covers all patients arriving in a shift---cluster at the shift level, analogous to a clustered randomized controlled trial. Importantly, as with the simple motivating example above, this line of reasoning would never justify clustering on examiners (as is sometimes done in practice).

\section*{A Checklist for Leniency Designs}

We now present a step-by-step guide to leniency designs, illustrated with the replication file from \textcite{FarreMensaHegdeLjungqvist2020data}. Their setting examines how initial patent approval by a US Patent Office examiner affects subsequent innovation by US startups. Our reanalysis sample contains 32,514 first-time patent applications filed after 2001, with final decisions by 2013 (see \Cref{sec:data} for details).

\Cref{tab:summary_stats} summarizes key variables. The treatment---approval of the first-time application---occurs for 65 percent of applicants.  Following \textcite{farre2020patent}, we track several key outcomes: whether startups file additional patents (40 percent do, per the table), whether they have any subsequent patents that were approved (30 percent do), and counts of new applications, approvals, and citations. Besides standard covariates like patent class and claim counts, we observe venture capital funding rounds prior to application. For the subset of applications also submitted abroad, we observe a proxy of application quality: European or Japanese patent office decisions. We also construct a measure of local entrepreneurship using the number of startups in each state. As we show below, these additional observables can help gauge the plausibility of the leniency design's identifying assumptions and assess external validity. See again \Cref{sec:data} for more details on these variables.

Our checklist consists of five steps; to illustrate it, we use an original R package for UJIVE, available at \href{https://github.com/kolesarm/ManyIV}{\texttt{https://github.com/kolesarm/ManyIV}}.
\medskip

\noindent\emph{1. Identify necessary controls for as-good-as-random assignment, and what estimator and standard errors are appropriate given the quasi-experimental design}

A credible leniency design begins with institutional knowledge that justifies quasi-random decision-maker assignment. This justification naturally identifies any required controls. In the US patent office, for example, once patent applications are allocated to art units the assignment process to individual examiners is not standardized. However, \textcite{LeSa12} argue, based on interviews with examiners, that the assignment process is effectively random conditional on the set of examiners working in the art unit at the patent office at the time of the assignment. Many art units use the last digit of the application serial number for assignment. Others similarly use rules that imply an effectively random assignment, such as a ``first-in-first-out'' rule that assigns each incoming application to the first available examiner. See also the institutional discussions in \textcite{SaWi19} and \textcite{farre2020patent}.

While the potential examiner set is not directly observable, we can assume the set of examiners working in a given art unit changes only slowly across time and specify the necessary controls as art unit-by-cohort (i.e., year) fixed effects. Without these necessary controls, the analysis would conflate systematic differences across art units or changes across cohorts with examiner-specific variation. From our discussion on heterogeneous effects, recall that it can be important to have sufficiently flexible controls to ensure a local average treatment effect interpretation of leniency design estimates.

Quasi-random assignment is only one piece of instrument validity in an examiner design; the second is the exclusion restriction, requiring decision-maker assignment to only affect relevant outcomes through the specified treatment. Here too, a clear argument should be made from institutional knowledge. In the patent setting example,  exclusion is plausible since patent examiners make relatively narrow approval decisions and likely have no direct impact on the future innovativeness of the applying firm.

The assignment mechanism can also guide the appropriate level for clustering standard errors.  Individual-level randomization in the patent setting, where patents are routed idiosyncratically to examiners within art unit and cohort, makes heteroskedasticity-robust standard errors appropriate. Based on our rule of thumb, clustering is not needed because each case is assigned independently. Contrast this with settings where groups of observations are assigned together: if all applications from a given month were assigned to a randomly chosen examiner, we would cluster by month. As discussed above, if clustering the standard errors is needed, this also necessitates using a leave-own-cluster-out version of the UJIVE estimator.

\medskip

\noindent\emph{2. Verify balance on other observables}

The as-good-as-random assumption implies that the average predetermined characteristics should be equal across decision-makers, conditional on the necessary controls. We test this implication directly: for each observable characteristic, we run our UJIVE specification with that characteristic as the outcome. Significant coefficients indicate imbalance and potential threats to identification.
This unified approach---using the same UJIVE specification for both balance tests and treatment effect estimation---offers important advantages over alternatives. Consider the common practice of regressing observables on constructed leniency measures from a first stage. With many examiners and fixed effects, this can generate mechanical correlations that masquerade as true imbalance.\footnote{Even if one uses a leave-out leniency measure, the coefficient will still be biased due to estimation error in the leniency measure, which generates an errors-in-variables bias. Another approach to testing balance is to regress the predetermined characteristics on the full set of examiner indicators and controls, and report the joint F-test for the examiner indicators. With many examiners, however, the default F-statistic is not valid \parencite{AnSo23}.} UJIVE avoids these finite-sample biases while maintaining the interpretability of our test: the magnitude of any imbalance directly translates to potential bias in our treatment effects. If venture capital funding has a positive coefficient in this UJIVE balance check, for example, this would imply that the randomly assigned examiners' tendency to approve patents is also positively correlated with patents' ex ante venture capital funding. Such a finding would raise questions about whether the examiners are truly randomly assigned, and the size of the UJIVE coefficient can help quantify sensitivity to omitted variables.

\Cref{tab:balance_checks} presents balance tests for the \textcite{farre2020patent} reanalysis.  Each row reports a UJIVE coefficient from regressing a predetermined covariate on patent approval, instrumenting with examiner indicators and controlling for art unit-by-year fixed effects. The results support quasi-random assignment: seven of eight coefficients are statistically insignificant at the 5 percent level, and all are insignificant at the 1 percent level. As importantly, the magnitudes are economically negligible: coefficients are typically 10 times smaller than the estimated treatment effects discussed below. Consider the venture capital variable, which might raise particular concern if certain examiners systematically reviewed applications from better-funded startups. We estimate a close-to-zero effect, suggesting no systematic differences (coefficient = -0.024, standard error = 0.035). Similar null results emerge for technical complexity (independent claims), external quality validation (European patent approval), and local entrepreneurial environment (state startup density). The absence of systematic imbalance strengthens our confidence that examiner assignment generates plausibly exogenous variation in patent approval.

Balance tests can also examine post-assignment variables to detect exclusion
restriction violations. These tests ask whether decision-makers affect outcomes
through channels beyond the specified treatment. In the patent setting, for instance, one concern is that examiner assignment affects not only whether an
application is approved but also the speed of the application review---which
plausibly has independent effects on follow-up innovation outcomes by reducing
applicant uncertainty. A UJIVE regression which uses a measure of review speed
as the outcome (e.g., months under review), keeping again the same treatment,
instrument, and controls, could be used to assess the significance of such
potential exclusion restriction violations: if present, the research design will
not generally recover the effect of patent-granting.\footnote{\label{fn:review_speed}When
  examiner assignment affects review speed, it is still possible to recover the
  effect of patent-granting by using review speed as a second treatment, as done
  in \textcite{farre2020patent} \parencite[see also][]{amms15}. One simply uses
  examiner dummies as instruments for
  both treatments in a UJIVE regression. This approach generally requires the
  effect of patent-granting to be constant, and the effect of review speed to be
  linear and constant (or at least not be selected on, see
  \cref{fn:selection_on_gains}). A more flexible approach is to interact approval with review time, discretized to a mutually exclusive set of dummies (e.g., under and over 2 years). This can again be estimated with a UJIVE
  regression. \Textcite{BhSi24} give conditions under which such
  multiple-treatment specifications have a causal interpretation under effect
  heterogeneity.}

\medskip

\noindent\emph{3. Estimate treatment effects by UJIVE and alternative estimators}

The UJIVE estimator is the preferred choice for leniency designs as it avoids the subtle potential biases of alternative approaches when there are many decision-makers and controls. Comparing UJIVE estimates and standard errors to those from these alternative approaches---such as 2SLS---is illustrative of these biases. As usual with instrumental variable designs, it can also be instructive to compare the UJIVE estimates to ordinary least squares estimates of the treatment effect as a way to probe the importance of selection bias.

\Cref{tab:treat_effect_est} presents treatment effect estimates across four specifications. Column 1 reports our preferred UJIVE estimates using examiner indicators as instruments with art unit-by-year controls. Columns 2 and 3 implement 2SLS with alternative instruments: the \textcite{farre2020patent} constructed leniency measure, based on examiner approval rates from earlier periods in Column 2, and the full set of examiner indicators in Column 3. Column 4 shows ordinary least squares results.

The UJIVE estimates reveal positive and significant effects of patent approval
on subsequent innovation. Startups with approved applications are more likely to file future
patents, receive future approvals, and generate citations, with significant
effects on both extensive and intensive margins. The pattern of bias across
estimators is instructive. Ordinary least squares produces higher coefficients
for subsequent applications but lower coefficients for approvals and citations,
suggesting moderate selection that operates differently across outcomes. The
2SLS estimates using examiner indicators (Column 3) fall between ordinary least
squares and UJIVE, pulled toward the biased ordinary least squares results. This
intermediate position confirms our theoretical prediction: with many
instruments, 2SLS suffers from finite-sample bias that partially reproduces the
selection problem it aims to solve. The UJIVE estimates avoid this bias,
providing our most credible evidence that initial patent approval causally
stimulates follow-on innovation. These effects operate through multiple
channels—approved firms not only attempt more patents but also succeed in
getting them approved and generating scientific impact through
citations.\footnote{\Cref{tab:treat_effect_est_w_controls,tab:complier_char_w_controls}
  in the appendix show how the estimates change when we additionally control for
  pre-application covariates. Adding such precision controls can reduce standard
  errors, as \cref{fn:precision_controls} discusses. This is not generally the
  case here, however, since the precision gains in the outcome equation are not sufficiently large to offset the additional first-stage noise the controls introduce.}
  The 2SLS estimates using the
\textcite{farre2020patent} measure of leniency tend to be \emph{larger} than
UJIVE\@. Since their leniency construction is similar to that of JIVE, this may
reflect a many-covariate bias.

Standard errors also tell an important story in \Cref{tab:treat_effect_est}. The many-instrument 2SLS (Column 3) produces standard errors roughly 3--4 times smaller than UJIVE (Column 1)—a difference that reflects statistical pathology rather than efficiency gains. As we have discussed, 2SLS standard errors and estimates are both pulled towards ordinary least squares with many examiners, creating a false sense of precision around the biased estimates. The shrinkage in standard errors occurs even when using the leave-out \textcite{farre2020patent} leniency measure in Column 2, since these do not reflect estimation error in the leniency measure. UJIVE avoids both failures: it delivers unbiased point estimates and standard errors that accurately capture sampling variation.

\medskip

\noindent\emph{4. Test monotonicity to assess the plausibility of a LATE interpretation}

To ensure that the UJIVE estimates have a clear causal interpretation under heterogeneous effects, we also need a version of first-stage monotonicity. Since the \textcite{ImAn94} first-stage monotonicity condition is likely too strong in the examiner setting, we implement a UJIVE test of the weaker average monotonicity assumption: that the average leniency of the examiners who would grant the firm a patent exceeds the average leniency of those who would deny it.

\Cref{fig:monotonicity} tests average monotonicity using outcome-specific UJIVE regressions. For each possible outcome value, we create an indicator for that value (e.g., a dummy for zero subsequent applications) and interact it with treatment status as our dependent variable, maintaining examiner instruments and art unit-by-year controls. The specification is estimated as in \Cref{tab:balance_checks,tab:treat_effect_est}. The fact that all of these UJIVE estimates are positive, with tight 95\% confidence intervals, builds support for average monotonicity holding.

\medskip

\noindent\emph{5. Characterize compliers to assess external validity}

Having found support for the design's internal validity, we next consider external validity. For this we estimate complier characteristics using UJIVE with modified outcomes: we regress the interaction of each covariate with the treatment indicator on the treatment itself, maintaining examiner instruments and controls. Such specifications identify weighted averages of complier characteristics for treated compliers using the same weights as our main estimates. Using one minus the treatment as the treatment variable identifies characteristics for untreated compliers, which should be the same by virtue of as-good-as-random assignment. We efficiently pool these two specifications, following \textcite{ANGRIST20231}.\footnote{For a characteristic $v_i$, one can show this can be done by using $\tilde{x}_i=2x_i-1$ as the transformed treatment variable that takes on values $\{-1,1\}$ instead of $\{0,1\}$. We then run a UJIVE regression of $v_i\times \tilde{x}_i$ onto $\tilde{x}_i$, maintaining the controls and examiner instruments.}

\Cref{tab:complier_char} reveals that compliers resemble the full sample. Column 2 reports complier means for eight characteristics, while Column 1 shows population means. None of the differences are statistically significant, while all the complier estimates fall within logical bounds---again supporting average monotonicity. Compliers have similar rates of venture capital funding, comparable technical complexity in their applications, and equivalent representation across patent classes. This representativeness matters for interpretation: our treatment effects likely approximate average effects for the full population of first-time patent applicants, not just the subset whose outcomes depend on examiner assignment. The validity of the UJIVE estimates thus appears to broadly inform how patenting affects startup innovation.

\section*{Conclusion}
Leniency designs allow for estimation of causal effects when expert decision-makers are as-good-as-randomly assigned and exert discretion on a high-stakes treatment. But, as with many things in both life and econometrics, the devil is in the details. Our review emphasizes how subtle choices in estimation and inference can have important consequences. The UJIVE estimator emerges as a natural solution to the many instrument and control problems that plague two-stage least squares estimation, while correctly estimated heteroskedastic-robust standard errors are a natural choice when assignment operates at the individual level. Our practical checklist and reanalysis of \textcite{farre2020patent} show how UJIVE can also be used for key auxiliary analyses, such as checking balance, testing average monotonicity, and characterizing compliers. We hope this toolkit gives applied researchers more confidence when approaching leniency designs, and we encourage them to consult this manual whenever questions of operation arise.

\singlespacing

\printbibliography
\clearpage

\begin{figure}[p]
\vspace{-20pt}  %
\caption{Monotonicity Checks}\label{fig:monotonicity}
\begin{subfigure}{\textwidth}
\caption{\# of subsequent applications}\label{fig:monotonicity_num_subseq_appl}
\includegraphics[width=\linewidth]{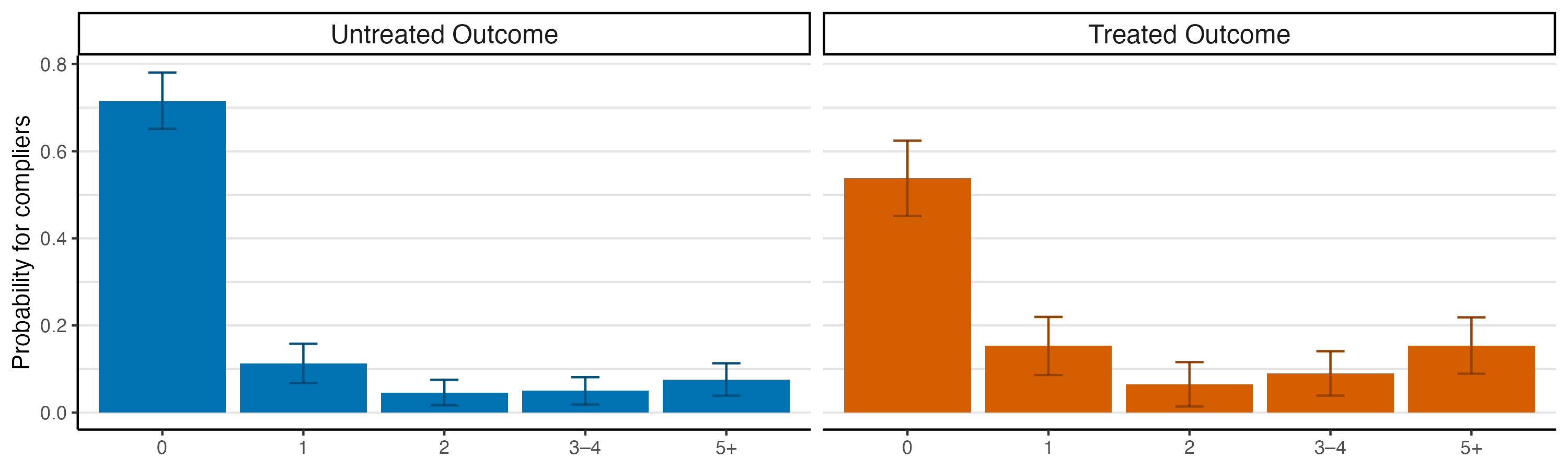}
\end{subfigure}\vspace{10pt}
\begin{subfigure}{\textwidth}
\caption{\# of subsequent approved applications}\label{fig:monotonicity_num_subseq_approval}
\includegraphics[width=\linewidth]{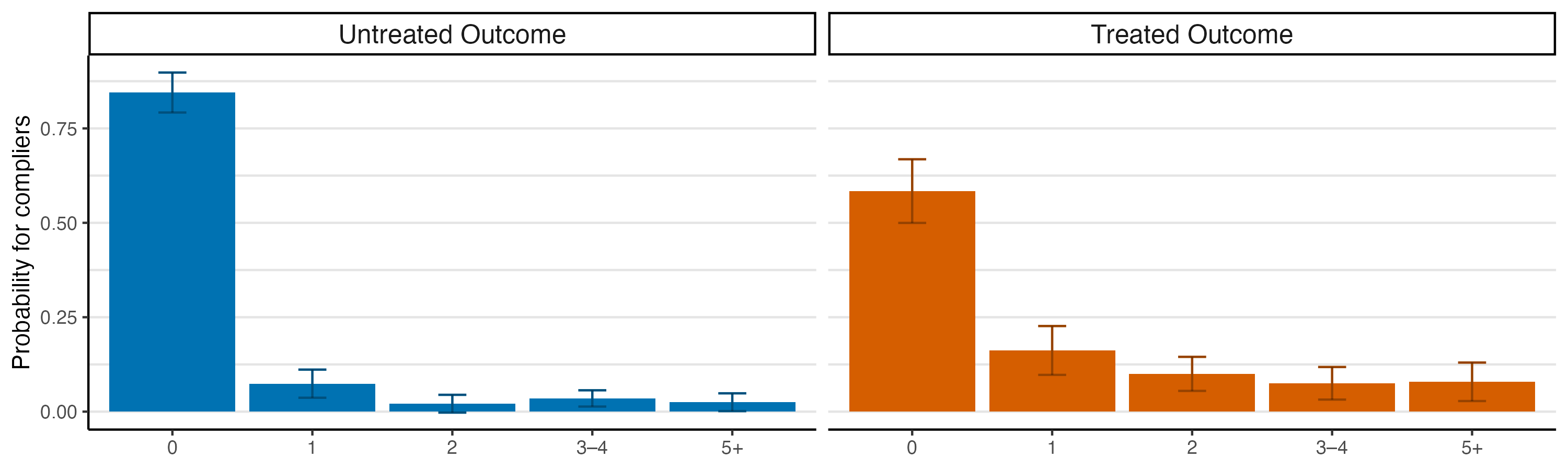}
\end{subfigure}\vspace{10pt}
\begin{subfigure}{\textwidth}
\caption{\# of citations to subsequent patents}\label{fig:monotonicity_num_subeq_cites}
\includegraphics[width=\linewidth]{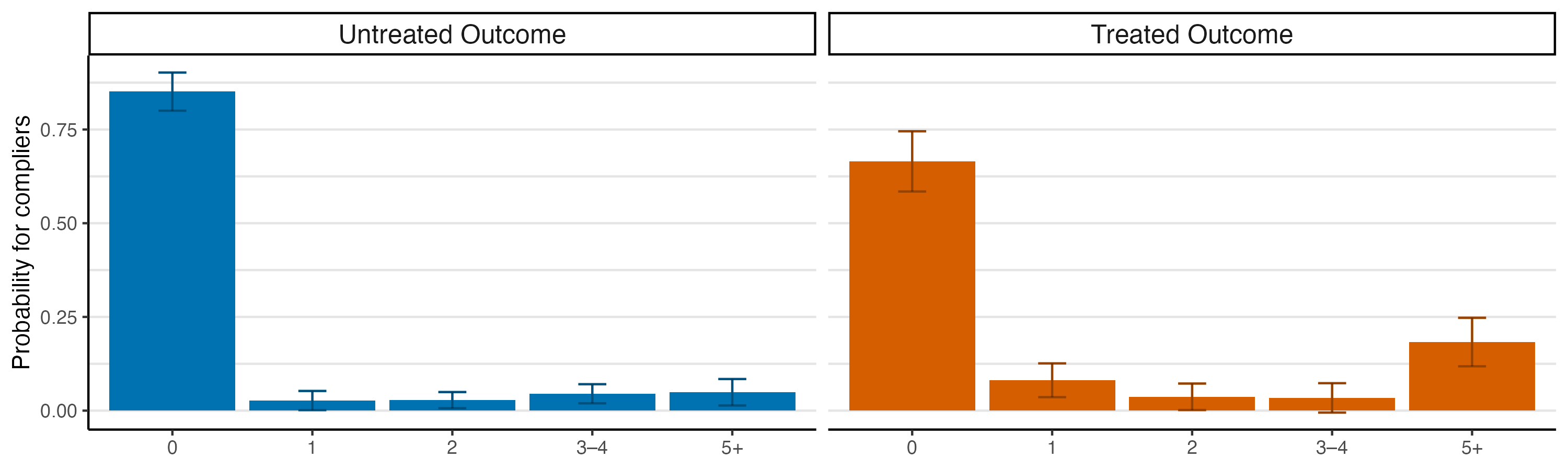}
\end{subfigure}
    \footnotesize
    Notes: This figure visualizes the distributions of treated and untreated outcomes of each outcome variable in Panel A of \Cref{tab:summary_stats}. Outcomes for treated and untreated compliers are estimated separately with UJIVE as described in the text. Vertical bars indicate 95\% confidence intervals that are robust to heteroskedasticity and treatment effect heterogeneity.
\vspace{-20pt}
\end{figure}

\clearpage

\begin{table}[p]
\vspace{-10pt}  %
\centering
\begin{threeparttable}
\footnotesize
\caption{\textcite{farre2020patent} Reanalysis Sample}\label{tab:summary_stats}
\begin{tabular}{l c c c}
\toprule
 & Mean & Std. Dev. & \# Obs. \\
 & (1) & (2) & (3) \\
\midrule
\multicolumn{4}{c}{\textit{Panel A\@: Outcomes}}\\[2pt]
Any subsequent application & 0.404 & 0.491 & 32,514\\
\# subsequent applications & 2.339 & 9.594 & 32,514\\
Any subsequent approved application & 0.299 & 0.458 & 32,514\\
\# subsequent approved applications & 1.276 & 6.021 & 32,514\\
Any citation to subsequent patents & 0.252 & 0.434 & 32,514\\
\# citations to subsequent patents & 5.786 & 52.808 & 32,514\\
\cmidrule(lr){1-4} \multicolumn{4}{c}{\textit{Panel B\@: Treatment}}\\[2pt]
Application approved & 0.649 & 0.477 & 32,514\\
\cmidrule(lr){1-4} \multicolumn{4}{c}{\textit{Panel C\@: Covariates}}\\[2pt]
Patent class group I & 0.162 & 0.368 & 32,514\\
Patent class group II & 0.410 & 0.492 & 32,514\\
Patent class group III & 0.445 & 0.497 & 32,514\\
\# independent claims in application & 3.730 & 3.174 & 27,226\\
\# VC rounds before application & 0.124 & 0.577 & 32,514\\
\# startups in HQ state & 317.3 & 354.3 & 32,514\\
Approval by European Patent Office & 0.372 & 0.484 & 3,287\\
Approval by Japanese Patent Office & 0.400 & 0.490 & 1,206\\
\bottomrule
\end{tabular}
    Notes: This table reports means and standard deviations for the sample used in reanalyzing \textcite{farre2020patent}. See \Cref{sec:data} for variable definitions.
\end{threeparttable}
\vspace{-10pt}
\end{table}

\begin{table}[p]
\centering
\begin{threeparttable}
\footnotesize
\caption{Balance Checks}\label{tab:balance_checks}
\begin{tabular}{l c c c}
\toprule
 & Coefficient & Std. Error & \# Obs. \\
 & (1) & (2) & (3) \\
\midrule
Patent class group I & -0.040 & 0.017 & 32,514\\
Patent class group II & 0.011 & 0.021 & 32,514\\
Patent class group III & 0.032 & 0.021 & 32,514\\
Log(\# independent claims in application) & 0.066 & 0.084 & 27,226\\
Log(1 + \# VC rounds before application) & -0.024 & 0.035 & 32,514\\
Log(\# startups in HQ state) & -0.273 & 0.143 & 32,514\\
Approval by European Patent Office & 0.013 & 0.212 & 3,287\\
Approval by Japanese Patent Office & 0.525 & 1.258 & 1,206\\
\bottomrule
\end{tabular}
    Notes: This table reports results of UJIVE regressions of each covariate in Panel C of \Cref{tab:summary_stats} on application approval, instrumenting with examiner indicators. All estimates control for art unit-by-year fixed effects. The reported standard errors are robust to heteroskedasticity and treatment effect heterogeneity.
\end{threeparttable}
\end{table}

\begin{table}[p]
\centering
\begin{threeparttable}
\footnotesize
\caption{Treatment Effect Estimates}\label{tab:treat_effect_est}
\begin{tabular}{l c c c c}
\toprule
 & UJIVE & \multicolumn{2}{c}{2SLS} & \multirow{3}{*}{OLS} \\
\cmidrule(lr){2-2}\cmidrule(lr){3-4} & Examiner & FMHL & Examiner  \\
 & indicators & approval rate & indicators & \\
 & (1) & (2) & (3) & (4) \\
\midrule
Any subsequent application & 0.173 & 0.265 & 0.232 & 0.234 \\
 & (0.055) & (0.023) & (0.016) & (0.006) \\
Log(1 + \# subsequent applications) & 0.323 & 0.456 & 0.374 & 0.357 \\
 & (0.100) & (0.037) & (0.027) & (0.009) \\
Any subsequent approved application & 0.259 & 0.250 & 0.240 & 0.223 \\
 & (0.050) & (0.020) & (0.014) & (0.005) \\
Log(1 + \# subsequent approved applications) & 0.356 & 0.362 & 0.323 & 0.291 \\
 & (0.081) & (0.029) & (0.021) & (0.007) \\
Any citation to subsequent patents & 0.183 & 0.210 & 0.173 & 0.164 \\
 & (0.049) & (0.020) & (0.014) & (0.005) \\
Log(1 + \# citations to subsequent patents) & 0.419 & 0.480 & 0.372 & 0.339 \\
 & (0.125) & (0.044) & (0.033) & (0.011) \\
\bottomrule
\end{tabular}
    Notes: This table reports estimates of the effect of application approval on each outcome in Panel A of \Cref{tab:summary_stats}. Column 1 estimates effects with UJIVE, instrumenting with examiner indicators. Columns 2 and 3 instead estimate effects with 2SLS, instrumenting with examiners' approval rate as computed by \textcite{farre2020patent} and examiner indicators respectively. Column 4 estimates effects by ordinary least squares. All estimates control for art unit-by-year fixed effects. The sample size is 32,514 in all specifications. Standard errors, reported in parentheses, are robust to heteroskedasticity and treatment effect heterogeneity.
\end{threeparttable}
\end{table}

\begin{table}[p]
\centering
\begin{threeparttable}
\footnotesize
\caption{Complier Characteristics}\label{tab:complier_char}
\begin{tabular}{l c c c}
\toprule
 & Sample  & Complier & \multirow{2}{*}{\# Obs.} \\
 & Mean & Mean &\\
 & (1) & (2) & (3) \\
\midrule
Patent class group I & 0.162 & 0.199 & 32,514\\
 &  & (0.031) & \\
Patent class group II & 0.410 & 0.415 & 32,514\\
 &  & (0.038) & \\
Patent class group III & 0.445 & 0.395 & 32,514\\
 &  & (0.036) & \\
\# independent claims in application & 3.730 & 3.660 & 27,226\\
 &  & (0.224) & \\
\# VC rounds before application & 0.124 & 0.158 & 32,514\\
 &  & (0.039) & \\
\# startups in HQ state & 317.3 & 329.3 & 32,514\\
 &  & (21.9) & \\
Approval by European Patent Office & 0.372 & 0.201 & 3,287\\
 &  & (0.097) & \\
Approval by Japanese Patent Office & 0.400 & 0.430 & 1,206\\
 &  & (0.518) & \\
\bottomrule
\end{tabular}
    Notes: This table reports means of each covariate in Panel C of \Cref{tab:summary_stats} for the full sample and for compliers. Complier means are estimated with UJIVE as described in the text. Standard errors, reported in parentheses, are robust to heteroskedasticity and treatment effect heterogeneity.
\end{threeparttable}
\end{table}

\clearpage

\begin{appendices}
\crefalias{section}{appsec}
\section{Derivations}\label{sec:derivations}
\subsection*{Estimation with Many Decision-makers}

To analyze the estimators in the main text, we follow the literature on many
instruments \parencite[e.g.,][]{bekker94,cshnw12}, and condition on the
realizations of the instruments and covariates, $\{z_{i}, w_{i}\}_{i=1}^{n}$,
throughout this section. To ease notation, we keep the conditioning implicit,
writing, e.g., $E[y_{i}]$ rather than $E[y_{i}\mid \{z_{i}, w_{i}\}_{i=1}^{n}]$.
To simplify the analysis, we also assume that the first-stage equation is
correctly specified, so that the residual $\nu_{i}$ is mean zero, and not just
uncorrelated with the instruments and covariates,
\begin{equation}\label{eq:fs_correct}
  E[\nu_{i}]=0.
\end{equation}
Since we are assuming that the instrument is as-good-as-randomly assigned
conditional on the covariates, it follows that $E[\varepsilon_{i}]$ depends only on
$w_{i}$ and not on $z_{i}$. To simplify derivations, we further assume that the
dependence is linear:
\begin{equation}\label{eq:structural_linear}
  E[\varepsilon_{i}]=w_{i}'\alpha
\end{equation}
for some vector $\alpha$. Similar to, for example, \textcite{cshnw12} we could
allow for asymptotically negligible non-linearities in both
\cref{eq:fs_correct,eq:structural_linear} at the cost of complicating the
derivations without affecting the results. Allowing for asymptotically
non-negligible non-linearities requires treating $z_{i}$ as random. Finally, we
assume that the errors $(\varepsilon_{i}, \nu_{i})$ are independent across $i$.

First consider the infeasible estimator $\hat{\beta}^{*}$. Since
$\tilde{\ell}_{i}$ depends only on the instruments and covariates, it follows
using~\cref{eq:fs_correct} that
$E[\sum_{i}\tilde{\ell}_{i}x_{i}]=\sum_{i}\tilde{\ell}_{i}\ell_{i}=\sum_{i}\tilde{\ell}_{i}^{2}$.
Similarly, using \cref{eq:structural_linear}, we have
$E[\sum_{i}\tilde{\ell}_{i}\varepsilon_{i}]=\sum_{i}\tilde{\ell}_{i}w_{i}'\alpha=0$,
where the last equality uses the fact that $\tilde{\ell}_{i}$ and $w_{i}$ are
orthogonal by definition of a sample residual,
$\sum_{i}\tilde{\ell}_{i}w_{i}=0$. Approximating the expectation of a ratio by a
ratio of expectations, we therefore obtain
\begin{equation*}
  E[\hat{\beta}^{*}]-\beta\approx \frac{0}{\sum_{i}\tilde{\ell}_{i}^{2}}=0.
\end{equation*}
 While using the ratio of expectations to
approximate the expectation of a ratio may appear ad hoc, one can show that the
approximation becomes exact asymptotically under the many instrument asymptotics
proposed by \textcite{bekker94}, where the dimension $K$ of the instrument
vector and the dimension $L$ of the covariate vector grow with sample size---the
results we give below mirror the asymptotic bias results given in
\textcite{ek17}, for instance. Therefore, an alternative interpretation of the
approximation $\approx$ is that it corresponds to the probability limit under
the \textcite{bekker94} sequence.

To derive the bias expressions for the other estimators, it is convenient to use
matrix notation. We denote the vectors of $n$ treatment and outcome observations
by $x$ and $y$, respectively, and denote the matrix of covariates and
instruments with rows ${w}_{i}'$ and $z_{i}'$ by $W$ and $Z$, respectively. To
state the results, it will be convenient to link the sum of squares of the
relative leniency, $\sum_{i}\tilde{\ell}_{i}^{2}$, to the first-stage $F$
statistic and the partial $R^{2}$.
 Recall that the first-stage $F$ statistic for testing the null that
$\pi=0$ is, under homoskedastic errors, given by
$F=\hat{\pi}'\tilde{Z}'\tilde{Z}\hat{\pi}/ K\var(\nu_{i})$, with expectation
$E[F]=\sum_{i}\tilde{\ell}_{i}^{2}/K\var(\nu_{i})+1$. The population partial
$R^{2}$ from adding the instruments to the first-stage regression is given by
$R^{2}=(\sum_{i}\tilde{\ell}_{i}^{2})/(\sum_{i}\tilde{\ell}_{i}^{2}+n\var(\nu_{i}))$
so
\begin{equation}\label{eq:F-R2}
\frac{\sum_{i}\tilde{\ell}_{i}^{2}}{\var(\nu_{i})}=  K(E[F]-1)=\frac{nR^{2}}{1-R^{2}}.
\end{equation}
This identity allows us to state the bias of a large class of estimators using
the first-stage $F$ statistic, as the following lemma shows.
\begin{lemma}\label{theorem:approx_bias}
  Consider an estimator $\hat{\beta}_{G}=y'Gx/x'Gx$ based on a relative leniency
  measure $Gx$ such that $Gx=\tilde{\ell}+G\nu$. Then, under homoskedastic
  errors,
  $E[\hat{\beta}_{G}]-\beta\approx \frac{\tr(G)\cov(\varepsilon_{i}, \nu_{i})}{
    \sum_{i}\tilde{\ell}_{i}^{2}+\tr(G) \var(\nu_{i})} = \frac{\tr(G)}{
    K(E[F]-1)+\tr(G)} \frac{\cov(\varepsilon_{i}, \nu_{i})}{\var(\nu_{i})} $.
\end{lemma}
The lemma follows by noting that
$E[x'Gx]=\sum_{i}\tilde{\ell}_{i}^{2}+E[\nu'G\nu]=
\sum_{i}\tilde{\ell}_{i}^{2}+\sum_{i}G_{ii}E[\nu_{i}^{2}]$, and
$E[\varepsilon'Gx]=\sum_{i}G_{ii}\cov(\varepsilon_{i}, \nu_{i})$. Under
homoskedastic errors, $E[\nu_{i}^{2}]$ and $\cov(\varepsilon_{i}, \nu_{i})$ do
not depend on the index $i$, which yields the first expression. The second
expression follows directly from \cref{eq:F-R2}.

The ordinary least squares estimator
$\hat{\beta}_{OLS}=\sum_{i}\tilde{x}_{i}y_{i}/\sum_{i}\tilde{x}_{i}^{2}$ uses
the relative leniency measure $Mx$, where $M=I-W(W'W)^{-1}W'$ denotes the $n\times n$
annihilator matrix associated with the covariates. Since the trace of $M$ is
$n-L$, where $L$ is the covariate dimension, it follows from
\Cref{theorem:approx_bias} that its bias is approximately
\begin{equation*}
  E[\hat{\beta}_{OLS}]-\beta\approx\frac{1-L/n}{
    K(E[F]-1)/n+1-L/n} \frac{\cov(\varepsilon_{i}, \nu_{i})}{\var(\nu_{i})}
  \approx (1-R^{2}) \frac{\cov(\varepsilon_{i}, \nu_{i})}{\var(\nu_{i})},
\end{equation*}
where the second approximation uses \cref{eq:F-R2} and the assumption that
$L/n$ is close to zero.

The 2SLS estimator uses the relative leniency measure $\tilde{z}_{i}\hat{\pi}$, which
in matrix notation may be written $Hx$, with
$H=\tilde{Z}(\tilde{Z}'\tilde{Z})^{-1}\tilde{Z}'$ denoting the projection matrix
associated with the instrument $\tilde{z}_{i}$. This matrix has trace equal to
$K$, which yields
\begin{equation*}
  E[\hat{\beta}_{TSLS}]-\beta\approx\frac{1}{
    E[F]} \frac{\cov(\varepsilon_{i}, \nu_{i})}{\var(\nu_{i})}
\end{equation*}
Hence, the relative bias of 2SLS, $(E[\hat{\beta}_{TSLS}]-\beta) \big/ (E[\hat{\beta}_{OLS}]-\beta)$, is approximately given by
$\frac{1}{(1-R^{2})E[F]}$, as claimed in the main text.

To show approximate unbiasedness of UJIVE, observe that we can write its
leniency measure using matrix notation as $G_{UJIVE}x$ with
$G_{UJIVE}=H-\diag(H_{ii}/(M_{ii}-H_{ii}))(M-H)$. Note that this matrix formula
also implies that UJIVE can be implemented in one step: one only needs to
compute the diagonals of the projection matrices $H$ and $M-H$, the fitted
values $Hx$, and the residuals $(M-H)x$, all of which obtain easily from
regressing $x_{i}$ onto $\tilde{z}_{i}$ and onto $z_{i},w_{i}$. The unbiasedness
property then follows directly from~\Cref{theorem:approx_bias} by noting that
$\tr(G_{UJIVE})=0$.

For the JIVE1 estimator studied in \textcite{aik99}, the relative leniency
measure can be written as $G_{JIVE}x$ with
$G_{JIVE}=M(I-D_{Q})^{-1}(H_{Q}-D_{Q})$, where $H_{Q}=H+W(W'W)^{-1}W'$ is the
projection matrix associated with the full vector of right-hand side variables
$(z_{i},w_{i})$, and $D_{Q}$ is its diagonal. Since
$\tr(G)=\tr((I-D_{Q})^{-1}(H-D_{Q}M))=-L$, we obtain
\begin{equation*}
  E[\hat{\beta}_{JIVE}-\beta]\approx -\frac{L}{
    K(E[F]-1)-L} \frac{\cov(\varepsilon_{i}, \nu_{i})}{\var(\nu_{i})},
\end{equation*}
which equals the 2SLS bias times $-\frac{E[F]}{K(E[F]-1)/L-1}$, as claimed in
the main text.

For IJIVE, the $G$ matrix takes the form
$G_{IJIVE}=M(I-\diag(H))^{-1}(H-\diag(H))M$, with trace equal to
$\tr((I-\diag(H))^{-1}(H-\diag(H)M))=\sum_{i}\frac{H_{ii}(1-M_{ii})}{1-H_{ii}}$.
While this quantity is always positive, under balanced design, i.e., when the
diagonal elements of $H_{ii}$ are all approximately equal, it simplifies to
$LK/(n-K)$, which is much smaller than $\tr(G_{JIVE})$.

To derive the bias-corrected 2SLS estimator, observe that the proof of
\Cref{theorem:approx_bias} shows that under homoskedasticity, the numerator and
denominator of the 2SLS formula have expectations
$\frac{1}{n}\sum_{i}E[\hat{\ell}_i
y_i]=\frac{1}{n}\sum_i\tilde{\ell}_i^2\beta+K\cov(\nu_i\beta+\varepsilon_i,
\nu_i)/n$, and
$\frac{1}{n}\sum_{i}E[\hat{\ell}_{i}
x_i]=\frac{1}{n}\sum_i\tilde{\ell}_i^2+K\var(\nu_i)/n$, respectively. An
unbiased estimator of $\var(\nu_i)$ obtains as the degrees-of-freedom adjusted
sample variance of the first-stage residuals, $\hat{\nu}=(M-H)x$, given by
$\widehat{\var}(\nu_i)=\sum_{i} \hat{\nu}_i^2/(n-K-L)=x'(M-H)x/(n-K-L)$. Since
$\nu_i\beta+\varepsilon_i$ corresponds to the reduced form residual from
regressing $y$ onto $Z$ and $W$, the degrees-of-freedom adjusted sample
covariance between the first-stage residuals and reduced-form residuals,
$(M-H)y$, given by $y'(M-H)x/(n-K-L)$, gives an unbiased estimator of
$\cov(\nu_i\beta+\varepsilon_i, \nu_i)$. Subtracting estimates of 2SLS numerator
and denominator bias gives the bias-corrected 2SLS estimator,
\begin{equation*}
\hat{\beta}_{B2SLS}=\frac{\sum_i\hat{\ell}_{i}y_i-Ky'(M-H)x/(n-K-L)}{\sum_i\hat{\ell}_{i}x_i-Kx'(M-H)x/(n-K-L)}
=\frac{y'Hx-Ky'(M-H)x/(n-K-L)}{x'Hx-Kx'(M-H)x/(n-K-L)},
\end{equation*}
which uses the leniency measure $G_{B2SLS}x$ with $G_{B2SLS}=H-K(M-H)/(n-K-L)$. This has trace equal to $0$, and $\hat{\beta}_{B2SLS}$ is therefore unbiased under homoskedasticity. This version of the bias-corrected 2SLS estimator corresponds to the version studied in \textcite{kcfgi14}.

Finally, we show that the UJIVE estimator and the FEJIV estimator of
\textcite{csw23} can be interpreted as minimizing the mean-squared error of the
estimated leniency under homoskedasticity. It follows from the proof of
\Cref{theorem:approx_bias} that a leniency measure $Gx$ is unbiased under
heteroskedasticity if the diagonal of $G$ is zero, $G_{ii}=0$, and if $GW=0$ and
$G\tilde{Z}=\tilde{Z}$ (the latter two conditions are equivalent to the
condition $Gx=\tilde{\ell}+G\nu$ in \Cref{theorem:approx_bias}). Under these
conditions, the mean squared error of the leniency measure is
$E[(Gx-\tilde{\ell})'(Gx-\tilde{\ell})]=E[\nu'G'G\nu]$. When $\nu_{i}$ is
homoskedastic, this expectation simplifies to $\var(\nu_{i})\tr(G'G)$.
Therefore, minimizing the mean squared error subject to an unbiasedness
constraint is equivalent to minimizing $\tr(G'G)$ subject to $\diag(G)=0$,
$GW=0$, and $G\tilde{Z}=\tilde{Z}$. The first-order condition for the Lagrangian
associated with this problem is given by
\begin{equation*}
  G'=W\Pi'+\tilde{Z}A'+\diag(\lambda),
\end{equation*}
where $\Pi$ is the matrix of the Lagrange multipliers associated with the
constraint $GW=0$, $A$ is the matrix of the Lagrange multipliers associated with
the constraint $G\tilde{Z}=\tilde{Z}$, and $\lambda$ is the vector of Lagrange
multipliers associated with the constraint $\diag(G)=0$. Multiplying the
first-order condition by $W'$ and $\tilde{Z}'$, respectively, allows us to solve
for $\Pi$ and $A$: $\Pi'=-(W'W)^{-1}W'\diag(\lambda)$ and
$A'=(\tilde{Z}'\tilde{Z})^{-1}\tilde{Z}'(I-\diag(\lambda))$. Plugging this back
into the first-order condition then yields $G'=H+(M-H)\diag(\lambda)$. Since the
diagonal of $G$ is zero, this implies $\lambda_{i}=-H_{ii}/(M_{ii}-H_{ii})$.
Thus, $G_{UJIVE}=H - \diag(H_{ii}/(M_{ii}-H_{ii}))(M-H)$ solves the minimization
problem: the UJIVE leniency measure minimizes the mean squared error subject to
an unbiasedness constraint.

Now consider the same minimization problem, but subject to the additional
constraint that $W'G=0$ and $\tilde{Z}'G=\tilde{Z}'$. The first additional
constraint implies that $W'Gx=0$, so that the leniency measure is orthogonal to
the covariates. The second additional constraint ensures that the in-sample
covariance of the leniency measure with $\tilde{Z}$ is the same as the in-sample
covariance of $x$ with $\tilde{Z}$: $\tilde{Z}'Gx=\tilde{Z}'x$. These additional
constraints imply that the resulting estimator corresponds to the MINQUE
estimator in \textcite{rao70}; they also ensure that the optimal $G$ matrix is
symmetric. Under these additional constraints, the first-order condition becomes
\begin{equation*}
  G'=W \Pi'+\tilde{\Pi}'W+\tilde{Z}A'+\tilde{A}'\tilde{Z}'+\diag(\lambda),
\end{equation*}
where $\tilde{\Pi}$ is the matrix of Lagrange multipliers associated with the
constraint $W'G=0$ and $\tilde{A}$ is the matrix of Lagrange multipliers
associated with the constraint $\tilde{Z}'G=\tilde{Z}'$. Solving for the Lagrange
multipliers as before, we obtain $G'=H-(M-H)\diag(\lambda)(M-H)$. Since the
diagonal of $G$ is zero, this implies
$\diag(H)=\diag((M-H)\diag(\lambda)(M-H))$, which is equivalent to
\begin{equation}\label{eq:fejive_system}
  \diag(H)=((M-H)\odot (M-H))\lambda,
\end{equation}
where $A\odot B$ denotes the Hadamard (elementwise) product of two matrices,
$(A\odot B)_{ij}=A_{ij}B_{ij}$. Provided that a solution $\lambda$ to the linear
system in~\cref{eq:fejive_system} exists (a sufficient, but not a necessary,
condition is that $(M-H)\odot (M-H)$ is invertible) we therefore obtain the
solution
\begin{equation*}
  G_{FEJIV}=H-(M-H)\diag(\lambda)(M-H),
\end{equation*}
which corresponds precisely to the FEJIV estimator of \textcite{csw23}. Because
the minimization imposes additional constraints, it follows that when $\nu_{i}$
is homoskedastic,
$E[(G_{UJIVE}x-\tilde{\ell})'(G_{UJIVE}x-\tilde{\ell})]\leq
E[(G_{FEJIV}x-\tilde{\ell})'(G_{FEJIV}x-\tilde{\ell})]$. Since the linear system
in \cref{eq:fejive_system} has dimension $n$, solving it may be challenging in
large datasets (if $n$ is in the tens or hundreds of thousands, which is common
in leniency applications). However, the additional constraint $W'G=0$ ensures that
the FEJIV estimator is invariant to adding linear functions of the covariates to
the outcome equation: this desirable property is not shared by UJIVE\@.

\subsection*{Heterogeneous Effects}

We now relax the assumption that treatment effects are constant and give a
simple statement of the local average treatment effect theorem of
\textcite{ImAn94}. As before, assume that there are $K$ examiners, with $k(i)$
denoting the examiner assigned to observation $i$, and $z_{i}$ denoting a
vector of indicators for examiner assignment. In a departure from the previous
section, we no longer condition on the instruments and covariates, but treat the
sample $\{y_i,x_i,z_i,w_i\}_{i=1}^n$ as \emph{iid}. Let
$\ddot{\ell}_{i}=\ddot{z}_{i}'\pi$ denote the population version of the relative
leniency, where $\ddot{z}_{i}$ is the population residual from regressing
$z_{i}$ onto $w_{i}$ (while the relative leniency measure
$\tilde{\ell}_{i}=\tilde{z}_{i}'\pi$ uses the sample residual, $\ddot{\ell}_{i}$
uses the population residual). Since $\ddot{z}_{i}$ depends only on $w_{i}$ and
the examiner assignment, we may write $\ddot{\ell}_{i}=\ddot{\ell}(k(i),w_{i})$.
We assume that the conditional mean of $z_{i}$ given $w_{i}$ is linear, so that
$E[\ddot{z}_{i}\mid w_{i}]=0$. However, we relax the assumption that the
first-stage is linear, so that $\ddot{\ell}(k, w_{i})$ is an approximation to,
but does not necessarily equal, the true relative leniency
$\ell^{*}(k,w_{i})-E[\ell^{*}(k(i),w_{i})]$, where $\ell^{*}(k,w)=E[x_{i}\mid
k(i)=k, w_{i}=w]$ denotes the true absolute leniency.

Using $\ddot{\ell}_{i}$ as an instrument in a population IV regression of
$y_{i}$ onto $x_{i}$ without any controls identifies
\begin{equation*}
  \beta^{*}=\frac{E[\ddot{\ell}_{i} y_{i}]}{E[\ddot{\ell}_{i} x_{i}]}.
\end{equation*}
The next result shows that $\beta^{*}$ can be written as a weighted average of simple IV
regressions restricted to the subpopulation with $w_{i}=w$ and $k(i)=k$ or
$k(i)=j$. These simple IV regressions identify
\begin{equation*}
  \beta(k,j,w)=\frac{E[y_{i}\mid k(i)=k,w_{i}=w]-E[y_{i}\mid k(i)=j,w_{i}=w]}{
    E[x_{i}\mid k(i)=k,w_{i}=w]-E[x_{i}\mid k(i)=j,w_{i}=w]}.
\end{equation*}

\begin{lemma}\label{theorem:leniency_iv_pairwise}
  Suppose that $E[\ddot{z}_{i}\mid w_{i}]=0$. Then leniency IV estimand
  $\beta^*$ may be written as a weighted average of pairwise IV regressions,
  \begin{equation}\label{eq:leniency_pairwise}
    \beta^*
    =\frac{\sum_{k>j}E[\omega(k,j,w_{i})\beta(k,j,w_{i})]}{
      \sum_{k>j}E[\omega(k,j,w_{i})]
    },
    \quad \omega(j,k,w)=
    p_{j}(w)p_{k}(w)(\ddot{\ell}(k,w)-\ddot{\ell}(j,w))(\ell^{*}(k,w)-\ell^{*}(j,w)).
  \end{equation}
\end{lemma}
\begin{proof}
  Recall the identity $E[(A_{i}-A_{i'})(B_{i}-B_{i'})]=2\cov(A_{i},B_{i})$,
  where $(A_{i},B_{i})$ is a pair of random variables, and $(A_{i'},B_{i'})$ is
  an independent copy. Let
  $\beta_{N}(k, j, w)=E[y_{i}\mid k(i)=k, w_{i}=w]-E[y_{i}\mid k(i)=j, w_{i}=w]$
  denote the numerator of $\beta(j, k, w)$. By iterated expectations,
  $E[\ddot{\ell}_{i}y_{i}]=E[\cov(\ddot{\ell}_{i},y_{i})\mid w_{i}]$, so by the covariance
  identity,
  \begin{equation*}
    \begin{split}
      E[\ddot{\ell}_{i}y_{i}]
      &=\frac{1}{2}E[E[(\ddot{\ell}_{i}-\ddot{\ell}_{i'})(y_{i}-y_{i'})\mid w_{i}=w_{i'}]]\\
      &=\frac{1}{2}E[E[(\ddot{\ell}(k(i),w_{i})-\ddot{\ell}(k(i'),w_{i}))\beta_{N}(k(i),k(i'),w_{i})\mid w_{i}=w_{i'}]]\\
      &=\frac{1}{2}\sum_{k,j}E[p_{j}(w_{i})p_{k}(w_{i})(\ddot{\ell}(k,w_{i})-\ddot{\ell}(j,w_{i}))\beta_{N}(k,j,w_{i})]\\
      &=\sum_{k>j}E[p_{j}(w_{i})p_{k}(w_{i})(\ddot{\ell}(k,w_{i})-\ddot{\ell}(j,w_{i}))\beta_{N}(k,j,w_{i})].
    \end{split}
  \end{equation*}
  An analogous argument applied to the denominator yields the result.
\end{proof}
The result shows that leniency IV identifies a weighted average of simple IV
regressions, where we restrict the sample to observations assigned to one of
two examiners in a given pair $(j,k)$, and restrict the covariates to equal a
particular value. If the first stage is correctly specified,
$\ddot{\ell}(k,w)-\ddot{\ell}(j,w)=\ell^{*}(k,w)-\ell^{*}(j,w)$, in which case
the weights are proportional to the product of the conditional assignment
probabilities to $j$ and $k$ times the squared differences in examiner leniency.

\Cref{theorem:leniency_iv_pairwise} is an algebraic result, and holds without
any assumptions on the model. To give a causal interpretation to $\beta^*$, we
need a causal interpretation for $\beta(k, j, w)$. To this end, if the
instruments are as-good-as-randomly assigned conditional on $w_i$, and an
exclusion restriction holds (but we do not impose monotonicity), Proposition 3
in \textcite{AnImRu96} shows that $\beta(k, j, w)$ is given by a weighted
difference between treatment effects for compliers (those who are treated when
assigned to $k$ but not when assigned to $j$), and defiers (those who are
treated when assigned to $j$ but not when assigned to $k$). The weight on
defiers is negative and equal to the proportion of defiers divided by the
leniency difference, while the weight on compliers is positive and equals the
proportion of compliers divided by the leniency difference:
\begin{equation}\label{eq:late}
\beta(k, j, w)=\frac{E[\beta_i C_i(k, j, w)]}{\ell^{*}(k,w)-\ell^{*}(j,w)},
\end{equation}
where $\beta_i$ is the individual treatment effect, and $C_i(k, j, w)$ denotes a
variable that equals $1$ if an observation is a complier and $-1$ if they are a
defier, so that $\ell^{*}(k,w)-\ell^{*}(j,w)=E[C_i(k, j, w)]$. Under monotonicity, there
are no defiers, so that $\beta(k, j, w)$ corresponds to the average treatment
effect for compliers. Combining this interpretation with
\Cref{theorem:leniency_iv_pairwise} then yields the version of the local average
treatment effect theorem we referred to in the main text.

The next result generalizes the local average treatment effect theorem to allow
for the presence of defiers. To state the result, let $x_i(k)$ denote the
potential treatment indicator that equals $1$ if examiner $k$ would treat
observation $i$, so that $C_i(k, j, w)=x_i(k)-x_i(j)$.
\begin{lemma}\label{lemma:lambda_weights}
  Suppose that $E[\ddot{z}_{i}\mid w_{i}]=0$ and that~\cref{eq:late} holds. Then
  $\beta^{*}=E[\lambda_{i}\beta_{i}]/E[\lambda_{i}]$, where
  \begin{equation*}
    \lambda_{i}=\cov(\ddot{\ell}(k(i), w_i), x_i(k(i))\mid i)=(\overbar{\ell}_{i}(1)-\overbar{\ell}_i(0))(1-\overbar{x}_{i})\overbar{x}_{i}.
  \end{equation*}
  Here $\overbar{x}_{i}=\sum_{k}p_{k}(w_{i})x_{i}(k)$ is the average treatment
  rate of unit $i$,
  $\overbar{\ell}_{i}(1)=\sum_{k}p_{k}(w_{i})\ddot{\ell}(k,w_{i})x_{i}(k)/\overbar{x}_i$
  is the average relative leniency of examiners who would treat $i$, and
  $\overbar{\ell}_{i}(0)=\sum_{k}p_{k}(w_{i})\ddot{\ell}(k,w_{i})(1-x_{i}(k))/(1-\overbar{x}_{i})$
  is the average leniency of those who wouldn't.
\end{lemma}
\begin{proof}
 By \Cref{theorem:leniency_iv_pairwise}, \cref{eq:late} and iterated
 expectations,
 \begin{equation*}
   E[\ddot{\ell}_{i}y_{i}]=
   \sum_{k>j}E[\omega(k,j,w_{i})\beta(k, j, w_{i})]
   =   \sum_{k>j}E[
   p_{j}(w_{i})p_{k}(w_{i})(\ddot{\ell}(k, w_{i})-\ddot{\ell}(j, w_{i}))
   (x_{i}(k)-x_{i}(j))\beta_{i}
   ]= E[\lambda_{i}\beta_{i}],
 \end{equation*}
 where
 $\lambda_{i}=\sum_{k> j}p_{j}(w_{i})p_{k}(w_{i})(\ddot{\ell}(k,w_{i})-\ddot{\ell}(j,w_{i}))
 (x_{i}(k)-x_{i}(j))$, while  $E[\ddot{\ell}_{i}x_{i}]=E[\lambda_{i}]$. Using the definitions of $\overbar{\ell}_{i}(1),\overbar{\ell}_{i}(0)$ and $\overbar{x}_i$, we have $\sum_{j}\ddot{\ell}(j,w_{i})p_{j}(w_{i})=\overbar{x}_i\overbar{\ell}_{i}(1)+(1-\overbar{x}_i)\overbar{\ell}_{i}(0)$, so that
 \begin{equation*}
   \begin{split}
     \lambda_{i}&=\frac{1}{2}\sum_{k, j}p_{j}(w_{i})p_{k}(w_{i})(\ddot{\ell}(k,w_{i})-\ddot{\ell}(j,w_{i}))(x_{i}(k)-x_{i}(j))\\
                &=\sum_{k, j}p_{j}(w_{i})p_{k}(w_{i})(\ddot{\ell}(k,w_{i})-\ddot{\ell}(j,w_{i}))x_{i}(k)\\
                &=\sum_{k}p_{k}(w_{i})\ddot{\ell}(k,w_{i})x_{i}(k)
                  -\sum_{j}\ddot{\ell}(j,w_{i})p_{j}(w_{i})\sum_{k}p_{k}(w_{i})x_{i}(k)\\
                &=\overbar{x}_{i}\overbar{\ell}_{i}(1)
                  -[\overbar{x}_{i}\overbar{\ell}_{i}(1)+(1-\overbar{x}_{i})\overbar{\ell}_i(0)]\overbar{x}_{i}.
   \end{split}
 \end{equation*}
 Note that by the covariance identity used in the proof of \Cref{theorem:leniency_iv_pairwise}, $\lambda_i=\cov(\ddot{\ell}(k(i), w_i), x_i(k(i))\mid i)$, which yields the result.
\end{proof}

\Cref{lemma:lambda_weights} shows that $\beta^*$ identifies a weighted average of individual treatment effects, with weights that are zero for non-responders---individuals whose treatment status doesn't depend on examiner assignment (for whom  $(1-\overbar{x}_{i})\overbar{x}_{i}=0$). A necessary and sufficient condition for the weights on the remaining individuals to be positive is that the relative leniency measure of examiners who would treat $i$ exceeds that of those who wouldn't, 
$\overbar{\ell}_{i}(1)\geq \overbar{\ell}_{i}(0)$. If the only covariate is an intercept, $w_i=1$, and the first-stage is correctly specified, the condition $\overbar{\ell}_{i}(1)\geq \overbar{\ell}_{i}(0)$ is equivalent to the average monotonicity condition in \textcite{frandsen2023judging},
$\cov(x_{i}(k(i)), \ell(k(i),1)\mid i)\geq 0$. This equivalence was noted in \textcite{sigstad_monotonicity_2025}. 
\Cref{lemma:lambda_weights} generalizes the setup in \textcite{frandsen2023judging} by allowing for covariates and
for the first-stage equation to be misspecified. Under misspecification, it may be the case that the average monotonicity 
condition doesn't hold, even if it does hold for the true relative leniency, $\ell^{*}(k,w_{i})-E[\ell^{*}(k(i),w_{i})]$.

\subsection*{Inference}

To derive a standard error for UJIVE that allows for treatment effect
heterogeneity, let
\begin{equation}\label{eq:rf}
  y_i = z_i'\pi_{Y} + w_i'\delta_{Y} + \nu_{Yi},
  \qquad E[\nu_{Yi}]=0
\end{equation}
denote the \emph{reduced form} regression of the outcome on instruments and
covariates. To avoid technical complications, as in the estimation section of the Appendix, we
condition on the instruments and covariates, and in analogy
to~\cref{eq:fs_correct}, we assume that the reduced form is correctly specified.
The parameter of interest is given by the estimand of the infeasible estimator
$\hat{\beta}^{*}$,
\begin{equation*}
  \beta^{*}=\frac{\sum_{i}E[\tilde{\ell}_{i}y_{i}]}{\sum_{i}E[\tilde{\ell}_{i}x_{i}]}
  =\frac{\sum_{i}\tilde{\ell}_{i}z_{i}\pi_{Y}}{\sum_{i}\tilde{\ell}_{i}^{2}}.
\end{equation*}
If treatment effects are constant, so that \cref{eq:structural_linear} holds,
then it follows from the outcome equation and~\cref{eq:rf} that
$\pi_{Y}=\pi\beta$, and $\beta^{*}=\beta$. Recall that the UJIVE estimator is
given by $\hat{\beta}_{UJIVE}=y'G_{UJIVE}x/x'G_{UJIVE}x$. Since
$G_{UJIVE}x=\tilde{\ell}+G_{UJIVE}\nu$, it follows from the first stage equation and \cref{eq:rf} that we may
decompose the numerator and denominator of the estimator as
\begin{equation*}
  \begin{pmatrix}
    y'G_{UJIVE}x\\
    x'G_{UJIVE}x
  \end{pmatrix}=
  \sum_{i}  \begin{pmatrix}
    \tilde{\ell}_{i}y_{i}\\
    \tilde{\ell}_{i}x_{i}
  \end{pmatrix} + 
  \sum_{i}\begin{pmatrix}
    r_{Y, i}\nu_{i}\\
    r_{i}\nu_{i}
  \end{pmatrix}+
  \sum_{i,j}\begin{pmatrix}
    G_{ij}\nu_{Y, i}\nu_{j}\\
    G_{ij}\nu_{i}\nu_{j}\\
  \end{pmatrix},
\end{equation*}
where $G_{ij}$ is the $(i, j)$ element of $G_{UJIVE}$, while
$r_{Y, i}=\sum_{j}G_{ji}(z_{j}'\pi_{Y} + w_{j}'\delta_{Y})$ and
$r_{i}=\sum_{j}G_{ji}(z_{j}'\pi + w_{j}'\delta)$ denote the ``signal'' in the
reduced form and the first stage. The second and third terms represent additional
noise components in the estimator relative to the infeasible estimator
$\hat{\beta}^{*}$. Since the third term is uncorrelated with the first two, it
follows that the covariance matrix of the left-hand side may be written
\begin{equation*}
  \Sigma=
  \sum_{i}\var\begin{pmatrix}
    \tilde{\ell}_{i}y_{i}+r_{Y, i}\nu_{i}\\
    \tilde{\ell}_{i}x_{i}+r_{i}\nu_{i}
  \end{pmatrix}+\sum_{i,j}
  \begin{pmatrix}
    G_{ij}^{2}\sigma^{2}_{y, i}\sigma^{2}_{x, j}+G_{ij}G_{ji}\sigma_{xy,i}\sigma_{xy,j}&    G_{ij}(G_{ij}+G_{ji})\sigma_{yx, i}\sigma^{2}_{x, j}\\
    G_{ij}(G_{ij}+G_{ji})\sigma_{yx, i}\sigma^{2}_{x, j}&G_{ij}(G_{ij}+G_{ji})\sigma^{2}_{x, i}\sigma^{2}_{x, j}
  \end{pmatrix}
\end{equation*}
where $\sigma^{2}_{y,i}=\var(\nu_{Y, i})$, $\sigma^{2}_{x,i}=\var(\nu_{i})$, and
$\sigma_{yx,i}=\cov(\nu_{Y, i}, \nu_{i})$. Furthermore, the numerator and
denominator can be shown to be asymptotically normal by a martingale central
limit theorem \parencite[e.g.,][Lemma D.5]{ek17} so that
\begin{equation}\label{eq:rf-normal}
  \begin{pmatrix}
    y'G_{UJIVE}x\\
    x'G_{UJIVE}x
  \end{pmatrix}\approx \mathcal{N}\left(
    \begin{pmatrix}
      \beta^{*}\sum_{i}\tilde{\ell}_{i}^{2}\\
      \sum_{i}\tilde{\ell}_{i}^{2}\\
    \end{pmatrix}
    , \Sigma\right).
\end{equation}
It then follows by an application of the delta method that in large samples,
$\hat{\beta}_{UJIVE}$ has an approximately normal distribution,
\begin{equation*}
  \hat{\beta}_{UJIVE}\approx \mathcal{N}(\beta^{*}, \sigma^{2}_{UJIVE}),
\end{equation*}
where
\begin{equation*}
\sigma^{2}_{UJIVE}=\frac{
\sum_{i}\var(\tilde{\ell}_{i}(y_{i}-x_{i}\beta^{*})+(r_{Y,i}-r_{i}\beta^*)\nu_{i})+
\sum_{i,j}(G_{ij}^{2}\sigma^{2}_{y-x\beta^{*}, i}\sigma^{2}_{x, j}+
G_{ij}G_{ji}\sigma_{y-x\beta^{*},x, i}\sigma_{y-x\beta^{*},x, j})
}{(\sum_{i}\tilde{\ell}_{i}^{2})^2}
\end{equation*}
In contrast, by the same arguments, the standard error for the infeasible
estimator is given by the square root of
$\sigma^{2}_{*}=\sum_{i}\var(\tilde{\ell}_{i}(y_{i}-x_{i}\beta^{*}))/(\sum_{i}\tilde{\ell}_{i}^{2})^2$.

The $(r_{Y,i}-r_{i}\beta^*)\nu_{i}$ term in the numerator arises due to variation in
complier treatment effects across complier groups. In particular, under
regularity conditions,
$\var((r_{Y,i}-r_{i}\beta^*)\nu_{i})\approx
\var(\tilde{z}_{i}'(\pi_{Y}-\pi\beta^{*})\nu_{i})$. Under constant treatment
effects, the reduced form coefficients are proportional to the first stage,
$\pi_{Y}=\pi\beta$, so that the first term in the above display may be replaced
by $\sum_{i}\var(\tilde{\ell}_{i}(y_{i}-x_{i}\beta))$, or equivalently
$\sum_{i}\var(\varepsilon_{i}\tilde{\ell}_{i})$. The last term in the numerator
is the \textcite{bekker94} many instrument term. It is negligible if the
instruments are strong enough in the sense that $E[F]$ is large.

Note that \Cref{eq:rf-normal} corresponds exactly to Equation~(3) in
\textcite{AngKo24} who consider an instrumental variables regression with a
single instrument, with the variation in the first stage leniency,
$\sum_{i}\tilde{\ell}_{i}^2$ playing the role of the first-stage regression
coefficient on the single instrument, and $\Sigma$ playing the role of the
covariance between the reduced form and the first stage. It then
follows from their analysis that delta-method based inference is not overly
optimistic even if the instruments are weak, provided that the correlation
between $(y-x\beta^*)'G_{UJIVE}x$ and $x'G_{UJIVE}x$, given by
\begin{equation*}
  \rho=\frac{\Sigma_{12}-\Sigma_{22}\beta^{*}}{\sqrt{\Sigma_{22}(\Sigma_{11}-2\beta^{*}\Sigma_{12}+{\beta^{*}}^{2}\Sigma_{22})}}
\end{equation*}
is not too large: as long as $|\rho|<0.76$, rejection rates for nominal 5\%
tests stay below 10\% regardless of the instrument strength. In contrast to the
single-instrument case, $\Sigma$ now contains additional terms relative to the
infeasible estimator, so that $\rho$ no longer maps to the
endogeneity parameter under homoskedasticity. Nonetheless, since $\Sigma$ is
consistently estimable, for any particular null hypothesis of interest, one can
plug in the hypothesized value of $\beta^{*}$ into the above expression along
with estimates of $\Sigma$ to verify that using the UJIVE standard errors
doesn't lead to overrejection. Alternatively, since $\rho$ is monotone in
$\beta^{*}$, bounding the treatment effect parameter yields bounds in the
plausible values of $\rho$: if these exclude large values of $\rho$, confidence
intervals based on UJIVE standard errors will have good coverage rates.

\section{Data Appendix\label{sec:data}}
The data in our application comes from \textcite{FarreMensaHegdeLjungqvist2020data}, the replication file for \textcite{farre2020patent}. Our reanalysis sample contains 32,514 first-time patent applications filed after 2001 with final decisions by 2013.\footnote{\label{fn:algorithm} Our analysis uses the full sample from the original paper as found in its replication package. The initial dataset has 34,435 observations. We drop 1 observation that is missing citation data. We then drop 1,851 observations with singleton covariates or instruments, and 69 observations with leverage of 1. This causes us to drop 378 collinear controls and 1,676 collinear IVs. See the code documentation  at \href{https://github.com/kolesarm/ManyIV}{\texttt{https://github.com/kolesarm/ManyIV}} for discussion on how this recursive algorithm is implemented to drop these observations and variables.}

The variables used in the analyses (and listed in \Cref{tab:summary_stats}) are defined as follows. \textit{\# subsequent applications} is the number of applications with a filing date greater than the first-action date of a firm's first application and \textit{any subsequent  application} equals one if any applications were made. \textit{\# subsequent approved applications}  is the number of approved applications with a filing date greater than the first-action date of a firm's first application and \textit{any subsequent approved application} equals one if any applications were approved. \textit{\# citations to subsequent patents} is the  number of citations received by all subsequent patent applications over the five years after each patent application's public disclosure date, and \textit{any citation to subsequent patents} equals one if any such citation was received.

Patent class groups are defined by related subject matter, as classified by the
US Patent Office (see
\url{https://www.uspto.gov/sites/default/files/patents/resources/classification/classescombined.pdf}).\\
\textit{Patent class group I} includes chemical and related arts, \textit{Patent
  class group II} includes communications, radiant energy, weapons, electrical,
and computer arts, and \textit{Patent class group III} includes material
science, mechanical manufacturing and power, and related arts. \textit{\#
  independent claims in application} counts independent claims made in the
patent application. \textit{\# VC rounds before application} is the number of
venture capital funding rounds the startup had secured prior to its first patent
application. \textit{\# startups in HQ state} counts the total number of startups
headquartered in the same state as the applicant in the same year as the
applicant. \textit{Approval by European Patent Office} and \textit{Approval by
  Japanese Patent Office} take value 1 if a patent was approved by a European
and Japanese patent office, respectively, take value 0 if an application was
made but not approved, and are missing if no application was made. \textit{Years
  from application to first action} counts the years between the patent
application and the first action on the patent application.

\begin{table}[tp]
\centering
\begin{threeparttable}
\footnotesize
\caption{Treatment Effect Estimates with Precision Controls}\label{tab:treat_effect_est_w_controls}
\begin{tabular}{l c c c c}
\toprule
 & UJIVE & \multicolumn{2}{c}{2SLS} & \multirow{3}{*}{OLS} \\
\cmidrule(lr){2-2}\cmidrule(lr){3-4} & Examiner & FMHL & Examiner  \\
 & indicators & approval rate & indicators & \\
 & (1) & (2) & (3) & (4) \\
\midrule
Any subsequent application & 0.203 & 0.278 & 0.250 & 0.250 \\
 & (0.064) & (0.025) & (0.016) & (0.006) \\
Log(1 + \# subsequent applications) & 0.401 & 0.478 & 0.413 & 0.385 \\
 & (0.116) & (0.039) & (0.029) & (0.010)  \\
Any subsequent approved application & 0.275 & 0.254 & 0.248 & 0.231 \\
 & (0.059) & (0.022) & (0.015) & (0.005) \\
Log(1 + \# subsequent approved applications) & 0.403 & 0.370 & 0.343 & 0.304 \\
 & (0.095) & (0.031) & (0.023) & (0.008)  \\
Any citation to subsequent patents & 0.216 & 0.220 & 0.187 & 0.176\\
 & (0.057) & (0.021) & (0.014) & (0.005)  \\
Log(1 + \# citations to subsequent patents) & 0.514 & 0.499 & 0.410 & 0.363\\
 & (0.146) & (0.047) & (0.035) & (0.012)  \\
\bottomrule
\end{tabular}
    Notes: This table reports estimates of the effect of application approval on each outcome in Panel A of \Cref{tab:summary_stats}. This table adds additional precision controls:  HQ state fixed effect; the number of independent claims in the application, along with a missing-value dummy, with the claims count set to 0 when missing and the dummy flags those observations; the number of VC rounds before the application, and the number of startups in the HQ state. Column 1 estimates effects with UJIVE, instrumenting with examiner indicators. Columns 2 and 3 instead estimate effects with 2SLS, instrumenting with examiners' approval rate as computed by \textcite{farre2020patent} and examiner indicators respectively. Column 4 estimates effects by ordinary least squares. All estimates control for art unit-by-year fixed effects. The sample size is 32,514 in all specifications. Standard errors, reported in parentheses, are robust to heteroskedasticity and treatment effect heterogeneity.
\end{threeparttable}
\end{table}

\begin{table}[tp]
\centering
\begin{threeparttable}
\footnotesize
\caption{Complier Characteristics with Precision Controls}\label{tab:complier_char_w_controls}
\begin{tabular}{l c c c}
\toprule
 & Sample  & Complier & \multirow{2}{*}{\# Obs.} \\
 & Mean & Mean &\\
 & (1) & (2) & (3) \\
\midrule
Patent class group I & 0.162 & 0.215 & 32,514\\
 &  & (0.037) & \\
Patent class group II & 0.410 & 0.416 & 32,514\\
 &  & (0.044) & \\
Patent class group III & 0.445 & 0.375 & 32,514\\
 &  & (0.042) & \\
\# independent claims in application & 3.730 & 3.622 & 27,226\\
 &  & (0.223) & \\
\# VC rounds before application & 0.124 & 0.178 & 32,514\\
 &  & (0.047) & \\
\# startups in HQ state & 317.3 & 343.0 & 32,514\\
 &  & (25.5) & \\
Approval by European Patent Office & 0.372 & 0.142 & 3,287\\
 &  & (0.111) & \\
Approval by Japanese Patent Office & 0.398 & 0.601 & 1,197\\
 &  & (0.479) & \\
\bottomrule
\end{tabular}
    Notes: This table reports means of each covariate in Panel C of \Cref{tab:summary_stats} for the full sample and for compliers. This table adds additional precision controls:  HQ state fixed effect; the number of independent claims in the application, along with a missing-value dummy, with claims count is set to 0 when missing and the dummy flags those observations; the number of VC rounds before the application, and the number of startups in the HQ state.  Complier means are estimated with UJIVE as described in the text. Standard errors, reported in parentheses, are robust to heteroskedasticity and treatment effect heterogeneity.
\end{threeparttable}
\end{table}

\end{appendices}

\end{document}